\documentclass[a4paper,12pt]{amsart}

\usepackage{thmtools,thm-restate} 
\usepackage{url}
\usepackage[english]{babel}
\usepackage[utf8]{inputenc}
\usepackage[margin=1in]{geometry}
\usepackage[colorinlistoftodos]{todonotes}

\usepackage{amssymb, amsmath, amscd, amsthm}
\usepackage{amsfonts}
\usepackage[all]{xy}
\usepackage{graphicx}

\usepackage[normalem]{ulem}

\newtheorem{cor}{Corollary}[section]
\newtheorem{lem}{Lemma}[section]
\newtheorem{prop}{Proposition}[section]
\theoremstyle{definition}
\newtheorem{defn}{Definition}[section]
\theoremstyle{remark}
\newtheorem{rem}{Remark}[section]

\numberwithin{equation}{section}
\newtheorem{ex}{Example}[section]

\newcommand{\R}{\mathbb{R}}
\newcommand{\N}{\mathbb{N}}
\newcommand{\Z}{\mathbb{Z}}
\newcommand{\p}{\varphi}
\newcommand{\id}{\mathrm{id}}

\newcommand{\eps}{{\varepsilon}}
\newcommand{\dnp}{d_\mathrm{NP}}
\newcommand{\db}{d_\mathrm{B}}
\newcommand{\dht}{d_\mathrm{HT}}
\newcommand{\di}{d_\mathrm{I}}
\newcommand{\epi}{\mathrm{epi}}

\newcommand{\veceps}{{\boldsymbol{\varepsilon}}}
\newcommand{\epsdiag}{{\vec{\boldsymbol{\varepsilon}}}}
\newcommand{\bsy}[1]{{\boldsymbol{#1}}}

\newcommand{\vecalpha}{{\boldsymbol{\alpha}}}
\newcommand{\vecbeta}{{\boldsymbol{\beta}}}
\newcommand{\veczero}{{\boldsymbol{0}}}
\newcommand{\alphadiag}{{\vec{\boldsymbol{\alpha}}}}
\newcommand{\bfp}{{\boldsymbol{\varphi}}}
\newcommand{\bfs}{{\boldsymbol{\psi}}}

\title{The Persistent Homotopy Type Distance}

\author{Patrizio Frosini}
\address{Patrizio Frosini, Department of Mathematics, University of
  Bologna, Italy}
\email{patrizio.frosini@unibo.it}
\author{Claudia Landi}
\address{Dipartimento di Scienze e Metodi dell’Ingegneria,
Universit\`a di Modena e Reggio Emilia
Reggio Emilia, Italy}
\email{claudia.landi@unimore.it}
\author{Facundo M\'emoli}
\address{Department of Mathematics,
The Ohio State University,
Columbus, Ohio, U.S.A.}
\email{memoli@math.osu.edu}

\date{\today}

\begin{document}
\maketitle

\begin{abstract}
We introduce the {\em persistent homotopy type distance} $\dht$ to compare two
real valued functions defined on possibly different homotopy
equivalent topological
spaces. The underlying idea in the definition of $\dht$ is to measure
the minimal  shift that is necessary to apply to one of the two
functions in order that the  sublevel sets of the two functions become
homotopy equivalent.  This distance is interesting in connection
with persistent homology. Indeed, our main result states  that $\dht$
still provides an upper  bound for the bottleneck distance between the
persistence diagrams of the intervening functions. Moreover,
because homotopy equivalences are weaker than homeomorphisms,  this
implies a lifting of the standard stability results provided by the $L^\infty$ distance and the
natural pseudo-distance $\dnp$. From a different standpoint, we prove
that $\dht$ extends  the $L^\infty$ distance and $\dnp$ in two
ways. First, we show that, appropriately restricting the category of
objects to which $\dht$ applies, it can be made to coincide with the other two distances.   Finally, we show that $\dht$ has an interpretation in terms of interleavings that naturally places it in the family of distances used in persistence theory.\\


\end{abstract}

\section{Introduction}
Persistent homology has been developed as a theory to study topological properties of noisy or incomplete data, establishing itself as a fundamental tool for topological data analysis \cite{EdMo13,Gr08,Ca09}. Persistent homology  is characterized by an invariant called the persistence diagram (also known as the barcode) which summarizes both topological features of a dataset and their prominence. One of the reasons for the success of persistent homology is that persistence diagrams change continuously provided that the input dataset also changes continuously. This is known as the Stability Theorem of Persistence \cite{edelsStab}. Usually, (a) changes in persistence homology are measured via the bottleneck distance between persistence diagrams, (b) datasets are modeled as real valued functions defined on the same space, and (c) one uses the $L^\infty$ distance between functions to quantify their changes:

\begin{restatable}[Stability Theorem of Persistence \cite{edelsStab}]{thm}{thmstab}\label{thm:stab}
Let $X$ be a compact polyhedron. Then, for all continuous tame functions $\varphi_1,\varphi_2:X\rightarrow \mathbb{R}$, and all integers $k\geq 0$,
\begin{equation}\label{eq:stab}\db(D_k(\varphi_1),D_k(\varphi_2))\leq \|\varphi_1-\varphi_2\|_\infty.\end{equation}
\end{restatable}

Above, $\db$ stands for the bottleneck distance between persistence diagrams, and $D_k(\varphi)$ is the persistence diagram corresponding to the $k$-th homology of the sub-level set filtration of the function $\varphi$.

In order to lift the Stability Theorem of Persistence to the case when functions are defined on different, albeit homeomorphic spaces, one can resort to the natural pseudo-distance. If two continuous  functions $\varphi_X:X\to \R$, $\varphi_Y:Y\to \R$ are given on homeomorphic spaces $X$ and $Y$, then the natural pseudo-distance $\dnp$ \cite{DoFr04}  between them is defined by
\begin{equation}\label{eq:dnp}
\dnp\left((X,\varphi_X),(Y,\varphi_Y)\right):=\inf_{h} \|\varphi_1-\varphi_2\circ h\|_\infty\end{equation}
 where $h$ varies in the set of all homeomorphisms from the topological space $X$ onto the topological space $Y$.
 If two continuous  functions $\varphi_X:X\to \R$, $\varphi_Y:Y\to \R$ are given on non-homeomorphic spaces $X$ and $Y$, then we set $\dnp\left((X,\varphi_X),(Y,\varphi_Y)\right):=\infty$.
Then, since persistence diagrams of sub-level set filtrations are invariant under reparameterization, one obtains an improvement of inequality (\ref{eq:stab}) stated in Theorem~\ref{thm:stab}:
\begin{equation}\label{eq:stab-np}\db(D_k(\varphi_X),D_k(\varphi_Y))\leq \dnp\big((X,\varphi_X),(Y,\varphi_Y)\big).\end{equation}

However, the natural pseudo-distance is not suitable when we are interested in analyzing functions defined on non-homeomorphic topological spaces.

In this paper we construct a new extended pseudo-metric,  called the \emph{persistent homotopy type distance}, denoted $\dht$, to quantify the distance between real-valued functions defined on different spaces which is meaningful when the spaces are at least homotopy equivalent. In plain words,  the persistent homotopy type distance is a generalization of  the natural pseudo-distance in that it uses homotopy equivalences in place of homeomorphisms. This allows us to use the persistent homotopy type distance to obtain a new and stronger stability theorem for persistent homology, which is the main contribution of our paper:

\begin{restatable}{thm}{thmstabht}\label{stability_theorem}
Let $X$ and $Y$ be compact polyhedra, and $k$ be any non-negative integer. Let $\varphi_X:X\rightarrow \R$ and $\varphi_Y:Y\rightarrow \R$ be continuous functions. Then,
$$\db\big(D_k(\varphi_X),D_k(\varphi_Y)\big)\leq \dht\big((X,\varphi_X),(Y,\varphi_Y)\big).$$
\end{restatable}
We note that in the statement above $\dht$ becomes infinity when the underlying spaces above are not homotopy equivalent.

\subsection*{Variations on the basic definition of $\dht$}
In the definition (\ref{eq:dnp}) of the natural pseudo-distance above, one can in fact consider two interesting variations: on one hand, one can extend the class of considered functions from scalar functions \mbox{$\p_X:X\to \R$} to vector-valued functions $\bfp_X=(\p_i):X\to\R^n$,  with  \mbox{$\|\bfp_X\|_\infty:=\sup_{x\in X}\max_{1\le i\le n}|\p_i(x)|$}; on the other hand, one can  restrict from the category $\mathbf{H_0}$ (whose objects are all topological spaces endowed with $\R^n$-valued continuous functions and morphisms are all homeomorphisms between topological spaces)  to any subcategory $\mathbf{H}$ of $\mathbf{H_0}$ closed under inverse \cite{BiFlFa08}: for any $\bfp_X:X\to \R^n$, $\bfp_Y:X\to \R^n$, 
\begin{equation}\label{eq:dnp_sub}
\dnp^{\mathbf{H}}\left((X,\bfp_X),(Y,\bfp_Y)\right):=\inf_{h\in {hom_{\mathbf{H}}(X,Y)}} \|\bfp_1-\bfp_2\circ h\|_\infty\end{equation}
if there is a homeomorphism from $X$ to $Y$ in  $\mathbf{H}$, $\dnp^{\mathbf{H}}\left((X,\bfp_X),(Y,\bfp_Y)\right)=\infty$ otherwise. 

Generalizing the natural pseudo-distance to vector-valued functions permits to lift the stability of the interleaving distance of multidimensional persistence modules \cite{Lesnick} in much the same way as we lift the stability of the bottleneck distance in one-dimensional persistence. 

Restricting the set of homeomorphisms allows for the application of the natural pseudo-distance to cases when the desired invariance is not the one expressed by any homeomorphism, as shown in previous papers  \cite{Fr14,FrJa16}. For example, two monotonic functions $\p,\psi:[0,1] \rightarrow \mathbb{R}$ with the same set of extrema are equivalent under $\dnp$ (and therefore equivalent for standard persistent homology) when every homeomorphism from $[0,1]$ to $[0,1]$ is accepted. Thus, suitably restricting the set of acceptable homeomorphisms would permit distinguishing two such functions.

From a different perspective, since the group of all self-homeomorphisms of a topological  space, even a compact one, is not itself compact, the possibility of restricting the set of homeomorphisms is also motivated by the desire of working with compact groups. This would be useful for obtaining interesting theorems e.g. good finite approximations of the considered groups.

Analogously to the ability to specify a subcategory $\mathbf{H}$ in the case of the natural pseudo-distance, our proposal for a persistent homotopy type distance also permits specifying what constitutes suitable classes of homotopy equivalences, therefore allowing to select the class that is judged more relevant for a given application.

\subsection*{The homotopy type distance as an interleaving distance}
Starting with  \cite{interleaving,lesnick-phd}, and more recently with \cite{Lesnick,Bubenik,interleaving-brief}, a unifying look at all the metrics usually used to state the stability theorems of persistence has been proposed in terms of interleaving distances. Interleavings apply between pairs of functors from the category of ordered reals to any other category. Interleavings are given by pairs of natural transformations between each one of the functors and a \emph{shifted} version of the other functor.

The interleaving distance measures the smallest shift that allows the existence of an interleaving.  Since many distances considered in topological data analysis can be formulated in terms of interleavings, it is natural to ask whether the same holds true for the  persistent homotopy type distance.  A further contribution of this paper is a positive answer to this question.  Related recent work in this direction is that of Blumberg and Lesnick \cite{lesnick-hid}. In a related thread, we prove that, when restricted to merge trees, our homotopy type distance agrees with the  interleaving distance between merge trees of Morozov et al. \cite{Morozov}.

\subsection*{Organization of the paper}
After introducing the persistent homotopy type distance in Section~\ref{sec:math-setting}, we discuss its properties and give some examples. In Section~\ref{sec:stability-main}, we prove that the bottleneck distance between persistence diagrams is upper-bounded by the persistent homotopy type distance (Theorem~\ref{stability_theorem}).  In other words, we lift  the Stability Theorem of Persistence  (Theorem~\ref{thm:stab})  to functions defined on different spaces provided that  they are homotopy equivalent. In Section~\ref{sec:mergetrees}, we show that the interleaving distance between merge trees can be obtained as a special case of the persistence homotopy type distance. Then, Section~\ref{sec:interleaving} offers an interpretation of the persistent homotopy type distance as an interleaving distance more in general. Finally, a discussion section offers some thoughts on possible extensions.

\section{Mathematical setting}\label{sec:math-setting}

For any integer $n\ge 1$, let us endow $\R^n$ with the partial order $\preceq$ defined by setting, for any $\vecalpha,\vecbeta\in\R^n$, $\vecalpha=(\alpha_i)\preceq \vecbeta=(\beta_i)$ whenever $\alpha_i\le \beta_i$ for $i=1,\ldots, k$.  When $\vecalpha\preceq \vecbeta$ we also write $\vecbeta\succeq\vecalpha$. For $\vecalpha=(\alpha_i)\in\R^n$, we set $\|\vecalpha\|_\infty=\max_{1\le i\le k}|\alpha_i|$. Recall that  for a function $\bfp_X:X\to \R^n$, we set $\|\bfp_X\|_\infty=\sup_{x\in X} \|\bfp_X(x)\|_\infty$.
For $\alpha\in\R$, we denote by  $\alphadiag$ the diagonal element $(\alpha,\alpha,\ldots, \alpha)\in\R^n$.

Let us consider the category $\mathbf{S}$  such that: the objects  of $\mathbf{S}$ are all the pairs $(X,\bfp_X)$ where $X$ is a  topological space and $\bfp_X:X\to\R^n$ is a vector-valued continuous  function; the morphisms of $\mathbf{S}$ from an object $(X,\bfp_X)$ to another object $(Y,\bfp_Y)$ are all the continuous maps $f: X\to Y$ such that $\bfp_Y\circ f\preceq \bfp_X$. The composition of morphisms is the usual composition of maps and identity morphisms are identity maps.

\begin{defn}
For every $\vecalpha\in\R^n$, with $ \vecalpha\succeq \veczero$,  the {\em $\vecalpha$-shift functor} $(\cdot)_\vecalpha : \mathbf{S}
\rightarrow \mathbf{S}$ is defined as follows:   for every $(X,\bfp_X)$ in $ob(\mathbf{S})$,  $(X,\bfp_X)_\vecalpha = (X,\bfp'_X)$, where $\bfp'_X(x)=\bfp_X(x)-\vecalpha$ for every $x\in X$;   for every $f:(X,\bfp_X)\to (Y,\bfp_Y)$ in $hom({\mathbf S})$, $(f)_\vecalpha=f$ regarded as a morphism between $(X,\bfp_X)_\vecalpha$ and  $(Y,\bfp_Y)_\vecalpha$.
\end{defn}

Instead of $\mathbf{S}$,  we can restrict ourselves to  any sub-category $\mathbf{C}$ of $\mathbf{S}$ provided that $\mathbf{C}$ is closed with respect to the {\em $\vecalpha$-shift functor} for any $\vecalpha\in\R^n$. From now on we assume one such $\mathbf C$ is fixed.

\begin{defn}\label{alpha-map}
Let $(X,\bfp_X)$, $(Y,\bfp_Y)$ be objects in $\mathbf{C}$. Given $\vecalpha\succeq \veczero$ in $\R^n$, an \emph{$\vecalpha$-map} $f$ with respect to $(\bfp_X,\bfp_Y)$ is a  morphism in $\mathbf{C}$ from $(X,\bfp_X)$ to $(Y,\bfp_Y)_\vecalpha$, that is a continuous map $f:X\to Y$ such that $\bfp_Y\circ f\preceq \bfp_X+\vecalpha$.  Furthermore, given two $\vecalpha$-maps $f_1:X\to Y$ and $f_2:X\to Y$, an \emph{$\vecalpha$-homotopy} in $\mathbf{C}$ between   $f_1$ and $f_2$ with respect to the pair  $(\bfp_X,\bfp_Y)$ is a continuous map $H:X\times [0,1]\to Y$  such that
\begin{enumerate}
\item $f_1\equiv H(\cdot,0)$;
\item $f_2\equiv H(\cdot,1)$;
\item $H(\cdot,t)$ is an \emph{$\vecalpha$-map} in $\mathbf{C}$ with respect to the pair
    $(\bfp_X,\bfp_Y)$ for every $t\in [0,1]$.
\end{enumerate}
If an \emph{$\vecalpha$-homotopy} between the $\vecalpha$-maps $f_1$ and $f_2$  with respect to the pair
    $(\bfp_X,\bfp_Y)$ exists, we say that  $f_1$, $f_2$ are \emph{$\vecalpha$-homotopic with respect to $(\bfp_X,\bfp_Y)$}.
\end{defn}

In plain words,  an $\vecalpha$-homotopy between $\vecalpha$-maps is a homotopy that is a $\vecalpha$-map at every instant.

\begin{rem}
Any map $f:X\to Y$ that induces a morphism  $f:(X,\bfp_X)\rightarrow (Y,\bfp_Y)_\vecalpha$ with respect to  $(\bfp_X,\bfp_Y)$, also induces a morphism  $f:(X,\bfp_X)\rightarrow (Y,\bfp_Y)_\vecbeta$, still with respect to  $(\bfp_X,\bfp_Y)$,  for every $\vecbeta\succeq \vecalpha$. Therefore, an $\vecalpha$-map can  also be regarded as  a $\vecbeta$-map, for every $\vecbeta\succeq \vecalpha$.
\end{rem}

\begin{rem}
For any fixed $\vecalpha\succeq \veczero$, the $\vecalpha$-homotopy with respect to  $(\bfp_X,\bfp_Y)$ is an equivalence relation on $\vecalpha$-maps from $(X,\bfp_X)$ to $(Y,\bfp_Y)_\vecalpha$.
\end{rem}

We now introduce the relation of $\vecalpha$-homotopy equivalence between objects of ${\mathbf C}$. In spite of its name, in general it is not going to be an equivalence relation.

\begin{defn}\label{alpha-eqiv}
For every $\vecalpha\succeq \veczero$ in $\R^n$ and any two objects $(X,\bfp_X)$ and $(Y,\bfp_Y)$  in ${\mathbf C}$, we say that $(X,\bfp_X)$ and $(Y,\bfp_Y)$ are {\em $\vecalpha$-homotopy equivalent} in ${\mathbf C}$ if there exist  $\vecalpha$-maps $f:X\to Y$ and $g:Y\to X$ in ${\mathbf C}$, with respect to $(\bfp_X,\bfp_Y)$ and $(\bfp_Y,\bfp_X)$ respectively, such that the following properties hold:
\begin{itemize}
\item the map
$g\circ f:X\to X$ is $2\vecalpha$-homotopic to $\id_X$ with respect to $(\bfp_X,\bfp_X)$; \item  the map
$f\circ g:Y\to Y$ is $2\vecalpha$-homotopic to $\id_Y$ with respect to $(\bfp_Y,\bfp_Y)$.
\end{itemize}
If the two previous conditions hold, we say that $g$ (resp. $f$) is an \emph{$\vecalpha$-homotopy  inverse} of $f$ (resp. $g$) with respect to $(\bfp_X,\bfp_Y)$ (resp. $(\bfp_Y,\bfp_X)$), and that $(f,g)$ constitutes a pair of \emph{$\vecalpha$-homotopy equivalences} in $\mathbf{C}$ with respect to the pair $(\bfp_X,\bfp_Y)$.
\end{defn}

 We observe that in general the $\vecalpha$-homotopy inverse of $f$ with respect to the pair $(\bfp_X,\bfp_Y)$ is not unique.

We are now ready to define  the persistent homotopy type distance.

\begin{defn}\label{defdG}
For each $\alpha\in\R$, set $\alphadiag=(\alpha,\alpha,\ldots  ,\alpha)\in \R^n$.  For each pair $\big((X,\bfp_X),(Y,\bfp_Y)\big)\in ob(\mathbf{C})\times ob(\mathbf{C})$   we set
$$\Lambda^\mathbf{C} \big((X,\bfp_X),(Y,\bfp_Y)\big):=\big\{\alpha\geq 0: \mbox{$(X,\bfp_X)$ and $(Y,\bfp_Y)$ are $\alphadiag$-homotopy equivalent in $\mathbf{C}$}\big\},$$
and
$$\dht^{\mathbf C}((X,\bfp_X),(Y,\bfp_Y)):=\inf\Lambda^\mathbf{C} \big((X,\bfp_X),(Y,\bfp_Y)\big),$$
where we use the convention that the infimum over the empty set is $+\infty$.

$\dht^{\mathbf C}$ will be called the \emph{persistent homotopy type (pseudo-)distance}  on the category $\mathbf{C}$. When ${\mathbf{C}}$ is taken to be the whole ${\mathbf{S}}$, sometimes we simply denote $\dht^{\mathbf C}$ by $\dht$.
\end{defn}

\begin{prop}\label{prop:homotopequiv}
Let $(X,\bfp_X)$ and $(Y,\bfp_Y)$  be two objects in $\bf S$ such that $\bfp_X$ and $\bfp_Y$ are bounded, that is $\|\bfp_X\|_\infty, \|\bfp_Y\|_\infty<+\infty$. Then, \mbox{$\dht((X,\bfp_X),(Y,\bfp_Y))<\infty$} if and only if  $X$ and $Y$ are homotopy equivalent. 
\end{prop}

\begin{proof}
Let us assume that $X$ and $Y$ are homotopy equivalent. Then, there exist maps $f:X\to Y$ and \mbox{$g:Y\to X$} and homotopies $H:X\times [0,1]\to X$ and $G:Y\times [0,1]\to Y$  between $g\circ f$ and $\id_X$  and between $f\circ g$ and $\id_Y$, respectively.  Since $\bfp_X$ and $\bfp_Y$ are bounded functions, the numbers
$\alpha_X:=\|\bfp_X-\bfp_Y\circ f\|_\infty$, $\alpha_Y:=\|\bfp_X\circ g-\bfp_Y \|_\infty$, $\widehat{\alpha}_X := \frac{1}{2}\max_{t\in[0,1]}\|\bfp_X-\bfp_X\circ H(\cdot,t)\|_\infty$, and $\widehat{\alpha}_Y:=\frac{1}{2}\max_{t\in[0,1]}\|\bfp_Y\circ G(\cdot, t)-\bfp_Y\|_\infty$ are all finite. Thus,
$\alpha:=\sup\{\alpha_X,\alpha_Y,\widehat{\alpha}_X,\widehat{\alpha}_Y\}<\infty,$
 and  $\alpha\in \Lambda^\mathbf{S} \big((X,\bfp_X),(Y,\bfp_Y)\big)$ so that $\dht\le \alpha$.
Conversely, if $\dht((X,\bfp_X),(Y,\bfp_Y))<\infty$, there is some  $\alpha\in \R$ such that  $(X,\bfp_X)$ and $(Y,\bfp_Y)$ are $\alphadiag$-homotopy equivalent in $\mathbf{S}$ with $\alphadiag=(\alpha,\alpha,\ldots,\alpha)$. Hence, $X$ and $Y$ are homotopy equivalent.
\end{proof}

In this paper we will follow the convention that $\alpha+\infty=\infty+\alpha=\infty$ for every $\alpha\in \overline{\mathbb{R}}:=\R\cup \{\infty\}$.

We now prove that $\dht^{\mathbf C}$ is an extended pseudo-metric. To this aim we use the following lemma.

\begin{lem}\label{lemma1}
 Let $(X,\bfp_X)$, $(Y,\bfp_Y)$, $(Z,\bfp_Z)$ be three objects in $\bf C$.
 If $(f_1,g_1)$ is a pair of $\vecalpha_1$-homotopy equivalences with respect to  $(\bfp_X,\bfp_Y)$ in $\bf C$ and
 $(f_2,g_2)$ is a pair of $\vecalpha_2$-homotopy equivalences with respect to  $(\bfp_Y,\bfp_Z)$ also in $\bf C$, then
 $(f_2\circ f_1, g_1\circ g_2)$ is a pair of $(\vecalpha_1+\vecalpha_2)$-homotopy equivalences with respect to the pair $(\bfp_X,\bfp_Z)$ in $\bf C$.
\end{lem}

 \begin{proof}
By definition,  $f_2\circ f_1$ (resp. $g_1\circ g_2$) is an $(\vecalpha_1+\vecalpha_2)$-map with respect to $(\bfp_X,\bfp_Z)$ (resp. $(\bfp_Z,\bfp_X)$).
Again by definition,
a $2\vecalpha_1$-homotopy $H_1:X\times [0,1]\to X$ with respect to $(\bfp_X,\bfp_X)$ from $g_1\circ f_1$ to $\id_X$ and
a $2\vecalpha_2$-homotopy $H_2:Y\times [0,1]\to Y$ with respect to $(\bfp_Y,\bfp_Y)$ from $g_2\circ f_2$ to $\id_Y$ exist. Thus, we can define a map
$\bar H:X\times [0,1]\to X$ by setting
$\bar H(x,t):= \begin{cases} g_1\circ H_2(f_1(x),2t), & \mbox{if } t\in [0,1/2) \\ H_1(x,2t-1), & \mbox{if } t\in [1/2,1] \end{cases}.$
 $\bar H$ is a $(2\vecalpha_1+2\vecalpha_2)$-homotopy with respect to $(\bfp_X,\bfp_X)$ from $g_1\circ g_2\circ f_2\circ f_1$ to $\id_X$.
Similarly, we can define a $(2\vecalpha_1+2\vecalpha_2)$-homotopy  $(\bfp_Z,\bfp_Z)$ from $f_2\circ f_1\circ g_1\circ g_2$ to $\id_Z$.
This proves the claim.
 \end{proof}

\begin{rem}
The term \emph{extended pseudo-metric} means that the function $\dht^{\mathbf C}$ is a function defined on $ob(\mathbf{C})\times ob(\mathbf{C})$ such that, for every $(X,\bfp_X),(Y,\bfp_Y) \in ob({\mathbf C})$, it holds that $(i)$ $\dht^{\mathbf C}((X,\bfp_X),(Y,\bfp_Y))\in [0,\infty]$, $(ii)$
if $(X,\bfp_X)=(Y,\bfp_Y)$ then $\dht^{\mathbf C}((X,\bfp_X),(Y,\bfp_Y))=0$, $(iii)$  $\dht^{\mathbf C}$ satisfies the symmetry property, $(iv)$  $\dht^{\mathbf C}$ satisfies the triangle inequality.
\end{rem}

\begin{prop}\label{prop_pseudometric}
The function $\dht^{\mathbf C}:ob(\mathbf{C})\times ob(\mathbf{C}) \rightarrow \overline{\mathbb{R}}$ is an extended pseudo-metric.
\end{prop}

\begin{proof}
We check that $\dht^{\mathbf C}$ satisfies the properties of an extended pseudo-metric.

\begin{enumerate}
\item 
By definition, $\dht^{\mathbf C}$ cannot take negative values.
\item $\dht^{\mathbf C}((X,\bfp_X),(X,\bfp_X))=0$ because the identity map of $X$, $\id_X$, belongs to $hom({\mathbf C})$ and $(\id_X,\id_X)$ is a pair of $\veczero$-homotopy equivalences with respect to $(\bfp_X,\bfp_X)$.
\item
The equality $\dht^{\mathbf C}((X,\bfp_X),(Y,\bfp_Y))=\dht^{\mathbf C}((Y,\bfp_Y),(X,\bfp_X))$ immediately follows from the symmetry of the definition of pairs of $\vecalpha$-homotopy equivalences.
\item Let $\bfp_X:X\to\R^n$, $\bfp_Y:Y\to\R^n$, $\bfp_Z:Z\to\R^n$ be three objects in our category $\mathbf{C}$. Let $\alphadiag=(\alpha,\alpha,\ldots,\alpha)\in\R^n$ with $\alpha\ge 0$. If either $(X,\bfp_X)$ is not $\alphadiag$-homotopy equivalent to $(Y,\bfp_Y)$, or  $(Y,\bfp_Y)$ is not $\alphadiag$-homotopy equivalent to $(Z,\bfp_Z)$ for any $\alpha\in\R$, then the definition of $\dht^{\mathbf C}$ implies that $\dht^{\mathbf C}((X,\bfp_X),(Y,\bfp_Y)) + \dht^{\mathbf C}((Y,\bfp_Y),(Z,\bfp_Z))=\infty$. In this case the inequality $\dht^{\mathbf C}((X,\bfp_X),(Y,\bfp_Y)) + \dht^{\mathbf C}((Y,\bfp_Y),(Z,\bfp_Z))\ge \dht^{\mathbf C}((X,\bfp_X),(Z,\bfp_Z))$ is trivially satisfied. Therefore, let us assume that $(f_1,g_1)$ is a pair of   $\alphadiag_1$-homotopy equivalences with respect to  $(\bfp_X,\bfp_Y)$ for some $\alpha_1\ge 0$, and $(f_2,g_2)$ is a pair of  $\alphadiag_2$-homotopy equivalences with respect to $(\bfp_Y,\bfp_Z)$ for some $\alpha_2\ge 0$. By definition of $\dht^{\mathbf C}$,
we can assume that
$\alpha_1\le \dht^{\mathbf C}((X,\bfp_X),(Y,\bfp_Y))+\eps$
and  $\alpha_2\le \dht^{\mathbf C}((Y,\bfp_Y),(Z,\bfp_Z))+\eps$ for an arbitrarily small $\eps>0$. We know from Lemma~\ref{lemma1} that
$(f_2\circ f_1, g_1\circ g_2)$ is a pair of $(\alphadiag_1+\alphadiag_2)$-homotopy equivalences with respect to $(\bfp_X,\bfp_Z)$.
It follows that
\begin{eqnarray*}
\lefteqn{\dht^{\mathbf C}((X,\bfp_X),(Z,\bfp_Z))\le \alpha_1+\alpha_2}\\
&\le& \dht^{\mathbf C}((X,\bfp_X),(Y,\bfp_Y))+ \dht^{\mathbf C}((Y,\bfp_Y),(Z,\bfp_Z))+2\eps.
\end{eqnarray*}
By taking the limit for $\eps$ tending to $0$, we obtain the triangle inequality
$$\dht^{\mathbf C}((X,\bfp_X),(Z,\bfp_Z))\le \dht^{\mathbf C}((X,\bfp_X),(Y,\bfp_Y))+ \dht^{\mathbf C}((Y,\bfp_Y),(Z,\bfp_Z)).$$
\end{enumerate}
\end{proof}

As we mentioned in the introduction,  our definition of the persistent homotopy type distance is meant to be a generalization of the natural pseudo-distance. The next proposition gives a relationship between these two distances when we compare functions on two homeomorphic spaces, whereas Example \ref{exNP} proves that we can read the natural pseudo-distance $\dnp$ into the persistent homotopy type distance $\dht^{\mathbf C}$ when we suitably restrict the underlying category.

\begin{prop}\label{rem-bound}
Let  $(X,\bfp_X)$ and $(Y,\bfp_Y)$ be objects in $\mathbf{S}$ where $X$ and $Y$ are homeomorphic. Then,
$$\dht((X,\bfp_X),(Y,\bfp_Y)) \leq \dnp\left((X,\bfp_X),(Y,\bfp_Y)\right).$$
 In particular, if $X=Y$, and $(X,\bfp_1),(X,\bfp_2)$ are objects in $\bf S$, then
$$\dht\big((X,\bfp_1),(X,\bfp_2)\big) \le\|\bfp_1-\bfp_2\|_{\infty}.$$
\end{prop}

\begin{proof}
If $\dnp\left((X,\bfp_X),(Y,\bfp_Y)\right)=\infty$ then there is nothing to prove. If not, then there exist $\alpha\geq 0$ and
a homeomorphism $f:X\to Y$ such that  $\|\bfp_X- \bfp_Y\circ f\|_\infty\le\alpha$. Thus, $f$ and $f^{-1}$ are  $\alphadiag$-maps, and  $f^{-1}\circ f$ and $f\circ f^{-1}$ are $\veczero$-homotopic, and hence $2\alphadiag$-homotopic, with $\alphadiag=(\alpha,\alpha,\ldots,\alpha)$  to $\id_X$ and $\id_Y$, respectively, by constant homotopies. Therefore, $(f,f^{-1})$ is a pair of  $\alphadiag$-homotopy equivalences for $(\bfp_X,\bfp_Y)$, implying that $(X,\bfp_X)$ and $(Y,\bfp_Y)$ are $\alphadiag$-homotopy equivalent. Thus, $\dht\le \dnp$.
\end{proof}

In the next few sections we will compare the persistent homotopy type distance with other metrics widely used in the topological data analysis literature to measure the perturbations in the input functions. In particular,  in Section~\ref{sec:interleaving} we will show that the persistent homotopy type distance can be represented as an interleaving type distance. For the moment, however, we  proceed with the study of the persistent homotopy type distance using the current  definition.

\subsection{Examples}
We now show how to compute the persistent homotopy type distance in some simple cases. We take $\mathbf{C}={\mathbf S}$ and denote $\dht^{\mathbf S}$ simply by $\dht$. For the sake of simplicity, we take $n=1$.

The first example pertains to the case when one can retract a given space $X$ to a subset $A$ in such a way that the function values do not increase.

\begin{prop}\label{prop:retract} Let $(X,\varphi_X)$ be an object in $\mathbf{S}$ and let $A$ be a subspace of $X$ such that there exists a  deformation retract $F:X\times[0,1]\rightarrow X$ of $X$ onto $A$ with the property that $\varphi_X(F(x,t))\leq \varphi_X(x)$ for all $x\in X$ and $t\in[0,1]$. Then,
$$d_{\mathrm{HT}}\big((X,\varphi_X),(A,\varphi_X|_A)\big)=0.$$
\end{prop}
\begin{proof}
The proof follows directly from the definition of the homotopy type distance.
\end{proof}
As an immediate corollary, we obtain:
\begin{ex}\label{prop:ex-contractible0}
Let $X$ be a contractible space and let $z\in X$. For fixed $c\in\R$, let $\p_c:X\to \R$ denote the constant function $\p_c(x)=c$ on $X$. Simply denote again by $c$ the constant function equal to $c$ on $\{z\}$.   It holds that $\dht((X,\p_c),(\{z\},c))= 0$.
\end{ex}

In  the next two examples, we  show that $\dht$ may be different from $\dnp$ even when the spaces are homeomorphic.

\begin{ex}\label{ex:band}  Let  $X$ be the band obtained by gluing without any twist two opposite sides  of a rectangle $R$, and  $Y$ the band obtained by gluing the same
sides of $R$ after applying a complete twist (i.e. a torsion  of $2\pi$ radians). Assume that the glued sides have length equal to 2 and that  $X$ and $Y$ are embedded into $\R^3$ as follows (see Figure~\ref{fig:band}):
\begin{figure}
\begin{center}\includegraphics[width=0.5\textwidth]{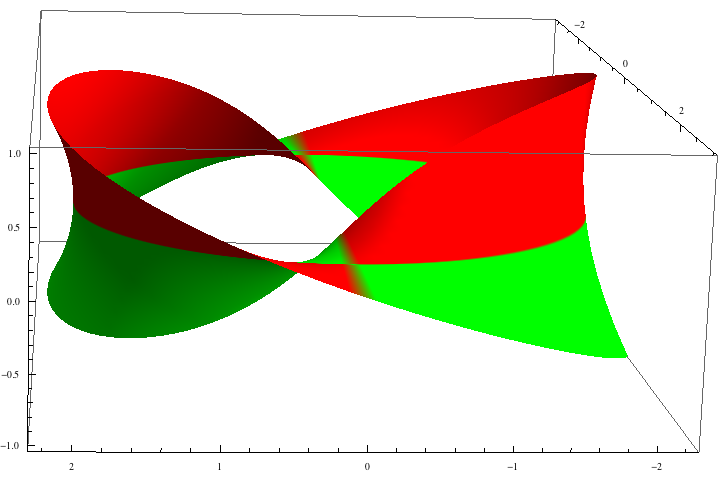}\end{center}
\caption{The twisted band $Y$ used in Example~\ref{ex:band}. Colors are as follows: the red region corresponds to points with non-negative $z$-coordinate whereas the green region is the complement, that is, it corresponds to those points with negative $z$-coordinate.}
\label{fig:band}
\end{figure}

$$X: \left\{\begin{array}{l}x(u,v)= 2\cos u\\
y(u,v)= 2\sin u\\
z(u,v)= v\end{array}\right.
\  \  \
Y: \left\{\begin{array}{l}x(u,v)= \left(2+v \sin u\right)\cos u\\
y(u,v)= \left(2+v \sin u\right)\sin u\\
z(u,v)= v\cos u\end{array}\right.$$
with $0\le u<2\pi$ and $-1 \le v \le 1$. Moreover, let $\p_X:X\to \R$ be defined by $\p_X(x,y,z)=z$ and  similarly $\p_Y:Y\to\R$  be defined by $\p_Y(x,y,z)=z$. This way  the centerlines of both $X$ and $Y$ coincide with the curve $C = \{(x, y, z) \in \R^3 : x^2 + y^2 = 4, z = 0\}$ and $\p_X$ takes values in  $\{-1,1\}$ at every boundary point of $X$, while $\p_Y$ continuously varies in $[-1,1]$ for every boundary point of $Y$. We claim that
\begin{equation}\label{eq:dht-ex}
\dht((X,\p_X),(Y,\p_Y))= 1,
\end{equation}
and
\begin{equation}\label{eq:dnp-ex}
\dnp((X,\p_X),(Y,\p_Y))=2.
\end{equation}


In order to establish (\ref{eq:dnp-ex}) note that for any homeomorphism $f:X\to Y$ one has that $\|\p_X-\p_X\circ f\|_\infty= 2$ simply because $f$ must take boundary points to
boundary points.
In order to establish (\ref{eq:dht-ex}), defining the retractions onto $C$ given by $r_X:X\to C$ by $r_X(x,y,z)=(x,y,0)$ and   $r_Y:Y\to C$ by $r_Y((2 + v \sin u) \cos u,(2 + v \sin u) \sin u,v \cos u)=(2\cos u,2\sin u,0)$, and the inclusions $i_{C,X}$ and $i_{C,Y}$ of $C$ into $X$ and $Y$ respectively,  one immediately checks that $(i_{C,Y}\circ r_X, i_{C,X}\circ r_Y)$ is a pair of  1-homotopy equivalences with respect to $(\p_X,\p_Y)$. This means that $\dht((X,\p_X),(Y,\p_Y))\leq 1$.
On the other hand, we claim that by Theorem \ref{stability_theorem} for $k=1$ we have that $\dht((X,\p_X),(Y,\p_Y))\geq 1$. Indeed, notice that $D_1(\varphi_X) = \{(-1,\infty)\}$ whereas $D_1(\varphi_Y) = \{(0,\infty)\},$ see Figure \ref{fig:band}.  It then follows that the bottleneck distance (see Section \ref{sec:stability}) between these diagrams satisfies $d_B(D_1(\varphi_X),D_1(\varphi_Y)) = 1$, which via  Theorem \ref{stability_theorem} implies our claim.

\end{ex}

\begin{ex}[Lens spaces]
Let $X=Y$ be the disjoint union of the lens spaces $L(7,1)$ and $L(7,2)$. Define $\p_X:X\to \R$ by setting $\p_{X|L(7,1)}\equiv 0$ and $\p_{X|L(7,2)}\equiv 1$, and define $\p_Y:Y\to \R$ by setting $\p_{Y|L(7,1)}\equiv 1$ and $\p_{Y|L(7,2)}\equiv 0$. Then, it holds that $\dht((X,\p_X),(Y,\p_Y))=0$  because $L(7,1)$ and $L(7,2)$ are homotopy equivalent but not homeomorphic \cite{SeifertThrelfall}, whereas $\dnp((X,\p_X),(Y,\p_Y))=1$.
\end{ex}

Finally, we consider another example which shows that $\dnp$ and $\dht$ can in fact be the same.
\begin{ex}
Let $M$ be any closed connected oriented manifold. Consider $X=Y=M$ and $\varphi_Y=\varphi_c$, the constant function equal to $c$. In this case $\dnp(\varphi_X,\varphi_c) = \|\varphi_X-c\|_\infty$. We claim that in this case $\dht(\varphi_X,\varphi_c)=\dnp(\varphi_X,\varphi_c)$. Assume that $\alpha\geq 0$ is such that there exists a pair $(f,g)$ a pair of $\alpha$-homotopy equivalences with respect to the pair $(\varphi_X,\varphi_c)$. We will prove that $\alpha\geq \|\varphi_X-c\|_\infty$. Note that from $c=\varphi_c(f(x))\leq \alpha + \varphi_X(x)$ for $x\in M$ we have that $\alpha\geq c-\min \varphi_X.$ Now, since $g\circ f$ and $\mathrm{id}_M$ are homotopic, their degrees are the same, therefore $|\mathrm{deg}(f)| = |\mathrm{deg}(g)|=1$. Hence, both $f$ and $g$ are surjective. From the condition $\varphi_X(g(x))\leq \alpha + \varphi_c(x) = \alpha+c$ for all $x\in M$ we obtain that $\alpha\geq \max \varphi_X-c.$ Hence, $\alpha\geq \max\big(c-\min \varphi_X,\max \varphi_X-c\big) = \|\varphi_X-c\|_\infty.$ The fact that $\dht\leq \dnp$ (Proposition \ref{rem-bound}) yields the claim.
\end{ex}

\subsection{Comments on Definition~\ref{defdG}}

We now consider three questions that may naturally arise:

\begin{itemize}
\item  What if we simplify the definition of $\alpha$-homotopies given in Definition~\ref{alpha-map} by removing the condition about $H$ being an $\alpha$-map at each instant, requiring only it to be an $\alpha$-map for $t=0$ and $t=1$? And, analogously, what if we remove the condition about an $\alpha$-homotopy in a subcategory ${\mathbf C}$ of ${\mathbf S}$ to be a morphism in ${\mathbf C}$ at each instant?
\item Would it be possible to define $\dht^{\mathbf C}$ via a minimum instead of an infimum?
\item Is $\dht^{\mathbf C}$ actually only an extended pseudo-metric or rather an extended metric?
\end{itemize}

We will answer this questions by means of examples,  taking $n=1$, i.e. real valued functions, for the sake of simplicity. Moreover, when possible, we take ${\mathbf C}={\mathbf S}$, and in such cases we simply write $\dht$ instead of $\dht^{\mathbf S}$.

As for the first issue, the definition of $\alpha$-homotopy we have given may seem more complex than necessary.  One could think of removing the condition about being an $\alpha$-map at each instant, maintaining only the condition that it be an $\alpha$-map for $t=0$ and $t=1$. Unfortunately, the new metric $d^*$ that we would obtain from this simplified definition of $\alpha$-homotopies would not give an upper bound for the bottleneck distance in persistent homology. In particular, the vanishing of
$d^*$ would not imply that the considered sublevel-set persistent homologies are the same. This is shown in the following example, proving that the analogue of Theorem~\ref{stability_theorem}  for $d^*$  would not hold.

\begin{figure}[htbp]
\begin{center}
\includegraphics[width=13cm]{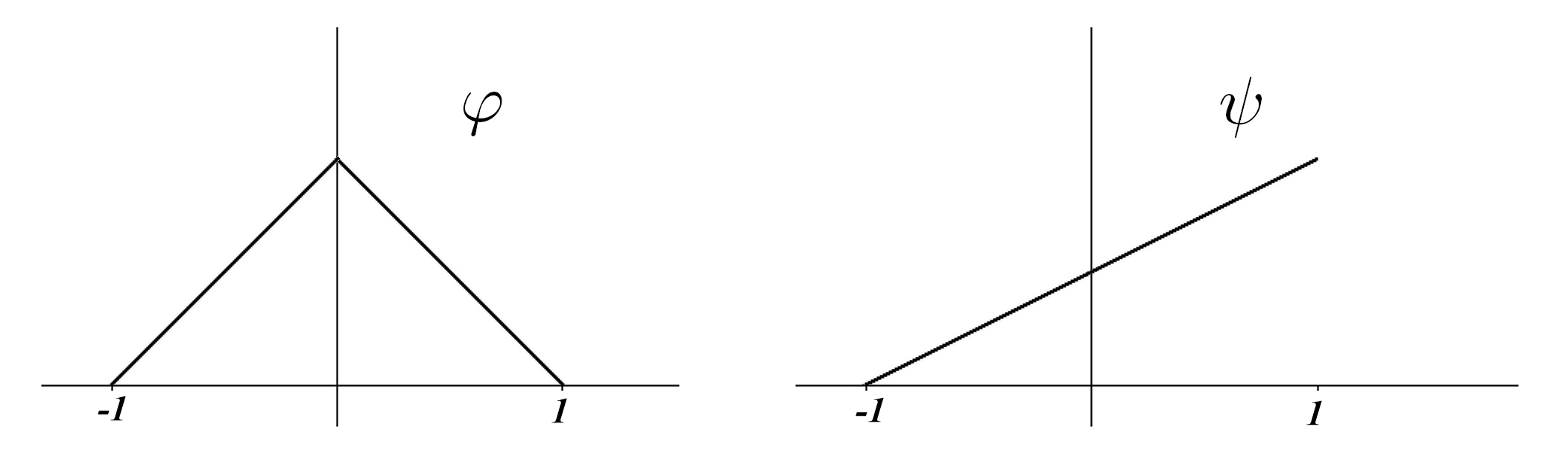}
\end{center}
\caption{The functions considered in Example~\ref{ex_d_star}. In this case $d^*(\varphi,\psi)=0$ but $\dht(([-1,1],\varphi),([-1,1],\psi))> 0$. }\label{fig1Sec2}
\end{figure}

\begin{ex}\label{ex_d_star}
 Define $\varphi,\psi:[-1,1]\to\R$ by  $\varphi(x)=1-|x|$ and $\psi(x)=(1+x)/2$ (see Figure~\ref{fig1Sec2}).
We also define $f,g:[-1,1]\to [-1,1]$ by setting $f(x)=1-2|x|$ and $g(x)=(x-1)/2$.
We have that $\psi \circ f(x)=\frac{1+(1-2|x|)}{2}=1-|x|=\varphi(x)$ and
$\varphi \circ g(x)=1-\left|\frac{x-1}{2}\right|=1-\frac{1-x}{2}=\frac{1+x}{2}=\psi(x)$, so that the maps $f$ and $g$ are $0$-maps with respect to $(\varphi,\psi)$ and $(\psi,\varphi)$, respectively. Furthermore,
$g \circ f(x)=\frac{(1-2|x|)-1}{2}=-|x|$ and
$f \circ g(x)=1-2\left|\frac{x-1}{2}\right|=1-2\left(\frac{1-x}{2}\right)=x$. It follows that  $f\circ g$ equals the identity, whereas $g \circ f$ is homotopic to the identity via the homotopy $H(x,t):=(t-1)|x|+tx$). Note that $H(\cdot,0)$ and $H(\cdot,1)$ are $0$-maps with respect to $(\varphi,\varphi)$.
As a consequence, $d^*(\varphi,\psi)=0$. Now we can observe that the sublevelset persistent homologies of $\varphi$ and $\psi$ are clearly different from each other, and hence the corresponding bottleneck distance is positive.
This can be easily seen by checking that $1$ is a homological critical value for $\varphi$, but not for $\psi$.
Therefore, $d^*$ is not an upper bound for the bottleneck distance.
On the contrary, we shall prove  in Section~\ref{sec:stability} that that property holds for $\dht$
(Theorem~\ref{stability_theorem}). This fact leads us to prefer the definition of $\alpha$-homotopy, and consequently of $\dht$, that we have presented.
\end{ex}

Analogously,  the following Example~\ref{comb} shows the difference between asking an $\alpha$-homotopy to be a morphism in $\mathbf{C}$ rather than in $\mathbf S$ at each instant.
\begin{ex}\label{comb}
Denoting by $2\Z$ the even integers and by $2\Z+1$ the odd integers, let $X=\{0\}\times\R\cup [0,1]\times 2\Z$ and $Y=\{0\}\times\R\cup [0,1]\times (2\Z+1)$ be two subsets of $\R^2$. Let $\p_X(s,t)=t$ and $\p_Y(s,t)=t$. Taking ${\mathbf C}$ the subcategory of $\mathbf S$ whose morphisms are homeomorphisms, it holds that $\dht^{\mathbf C}((X,\p_X),(Y,\p_Y))=1$ whereas $\dht((X,\p_X),(Y,\p_Y))=0$. Notice that removing the request for an $\alpha$-homotopy in $\mathbf C$ to be a homeomorphism at each instant, we would obtain $\dht^{\mathbf C}((X,\p_X),(Y,\p_Y))=0$ by using the $0$-maps $f:X\to Y$ $f(s,t)=(s,t-1)$ and $g:Y\to X$, $g(s,t)=(s,t-1)$ shown in Figure~\ref{fig2Sec2}.
\end{ex}
\begin{figure}[htbp]
\begin{center}
\includegraphics[height=4cm]{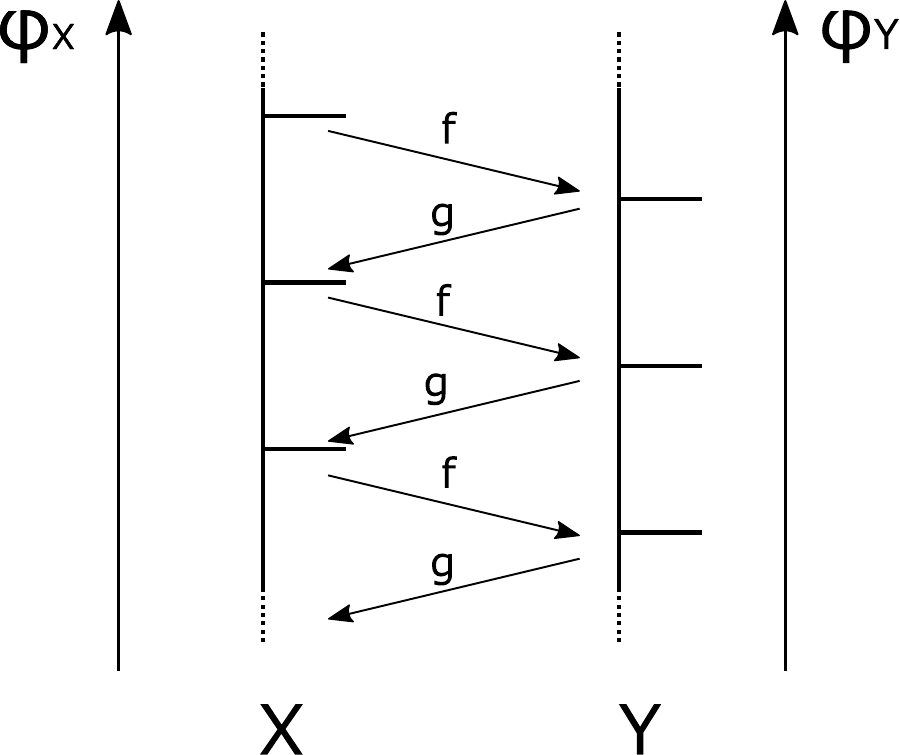}
\end{center}
\caption{The objects and maps considered in Example~\ref{comb}.}\label{fig2Sec2}
\end{figure}

As for the second issue, the use of an infimum instead of a minimum  is necessary, as the following example shows.
The same example shows that $\dht$ is not an extended metric, but only an extended pseudo-metric, thus clarifying the third issue.

\begin{figure}[htbp]
\begin{center}
\includegraphics[width=13cm]{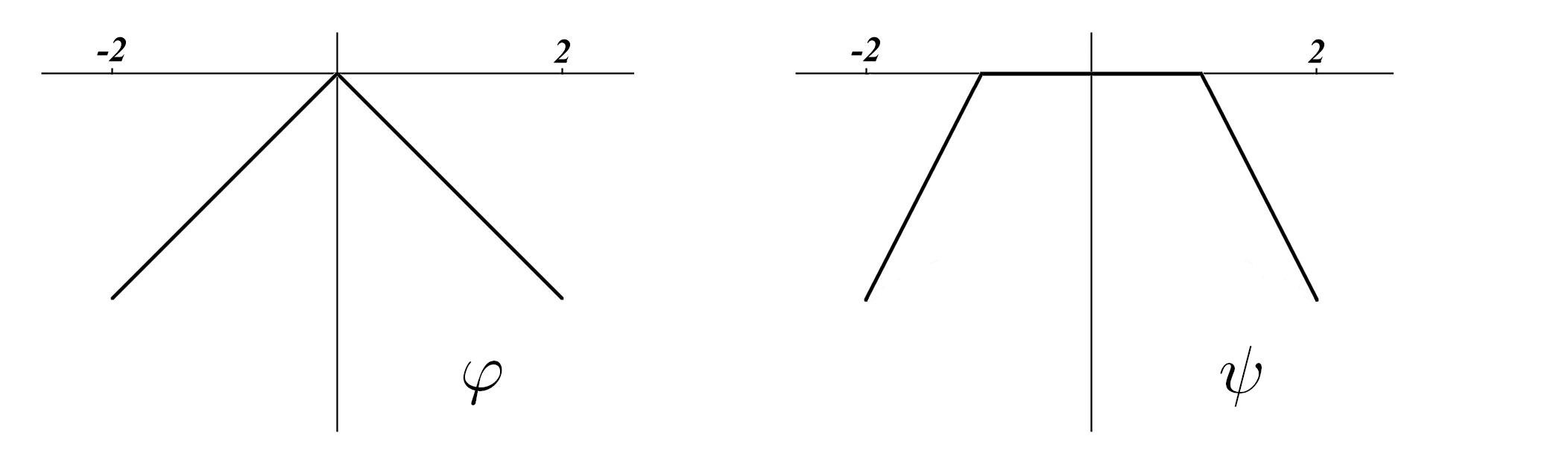}
\caption{The functions considered in Example~\ref{ex_inf}. In this case the infimum of the values $\alpha>0$ such that a pair of $\alpha$-homotopy equivalences with respect to $(\varphi,\psi)$ exists is $0$. However, no pair of $0$-homotopy equivalences with respect to $(\varphi,\psi)$ exists.}\label{fig3Sec2}
\end{center}
\end{figure}

\begin{ex}\label{ex_inf}
Define $\varphi,\psi:[-2,2]\to\R$ by setting $\varphi(x)=-|x|$, $\psi(x)=0$ for $|x|\le 1$ and $\psi(x)=2-2|x|$ for $|x|>1$  (see Figure~\ref{fig3Sec2}).
We claim that, for every small enough $\alpha>0$ in $\R$, we can find a pair of $\alpha$-homotopy equivalences $(f_\alpha,f_\alpha^{-1})$ with respect to $(\varphi,\psi)$. As a consequence, $\dht(([-2,2],\varphi),([-2,2],\psi))=0$. In order to show our claim, let us take $0<\alpha<2$ and consider the homeomorphism $f_\alpha:[-2,2]\to [-2,2]$  defined by setting $$f_\alpha(x):=
\begin{cases} \frac{1}{2-\alpha}x+\frac{2\alpha-2}{2-\alpha}, & \mbox{if } -2\le x< -\alpha \\
\frac{1}{\alpha}x, & \mbox{if }  -\alpha\le x< \alpha \\
\frac{1}{2-\alpha}x-\frac{2\alpha-2}{2-\alpha}, & \mbox{if }  \mbox{if } \alpha\le x\le 2.
\end{cases}$$
Observe that $f_\alpha([-2,-\alpha])=[-2,-1]$, $f_\alpha([-\alpha,\alpha])=[-1,1]$ and $f_\alpha([\alpha,2])=[1,2]$.
It follows that
$$\psi\circ f_\alpha(x):=
\begin{cases} 2-2\left|\frac{1}{2-\alpha}x+\frac{2\alpha-2}{2-\alpha}\right|=2+2\left(\frac{1}{2-\alpha}x+\frac{2\alpha-2}{2-\alpha}\right)=2\frac{x+\alpha}{2-\alpha}, & \mbox{if } -2\le x< -\alpha \\
0, & \mbox{if }  -\alpha\le x< \alpha \\
2-2\left|\frac{1}{2-\alpha}x-\frac{2\alpha-2}{2-\alpha}\right|=2-2\left(\frac{1}{2-\alpha}x-\frac{2\alpha-2}{2-\alpha}\right)=-2\frac{x-\alpha}{2-\alpha}, & \mbox{if } \alpha\le x\le 2.
\end{cases}
$$
As a consequence
$$\varphi(x)-\psi\circ f_\alpha(x):=
\begin{cases} -|x|-2\frac{x+\alpha}{2-\alpha}=x-2\frac{x+\alpha}{2-\alpha}=-\alpha\frac{2+x}{2-\alpha}, & \mbox{if } -2\le x< -\alpha \\
-|x|, & \mbox{if }  -\alpha\le x< \alpha \\
-|x|+2\frac{x-\alpha}{2-\alpha}=-x+2\frac{x-\alpha}{2-\alpha}=-\alpha\frac{2-x}{2-\alpha}, & \mbox{if } \alpha\le x\le 2.
\end{cases}
$$
We can easily check that $-\alpha\le \varphi(x)-\psi\circ f_\alpha(x)\le 0$ for every $x\in [-2,2]$.
It follows that $\|\varphi-\psi\circ f_\alpha\|_\infty\le \alpha$.
Hence, $(f_\alpha,f^{-1}_\alpha)$ is a pair of $\alpha$-homotopy equivalences with respect to $(\varphi,\psi)$. Given that $\alpha$ can be chosen arbitrarily close to 0, this implies that $\dht(([-2,2],\varphi),([-2,2],\psi))=\dnp(([-2,2],\varphi),([-2,2],\psi))=0$.

We claim that no pair of $0$-homotopy equivalences with respect to $(\varphi,\psi)$ exists. We  prove this by contradiction. Assume that a pair of $0$-homotopy equivalences $(f_0,g_0)$ with respect to $(\varphi,\psi)$ exists. Then a homotopy $H_0:[-2,2]\times [0,1]\to[-2,2]$ exists, such that $H_0(x,0)=g_0\circ f_0(x)$, $H_0(x,1)=x$ and $\varphi(H_0(x,t))\le \varphi(x)$ for every $x\in [-2,2]$ and every $t\in[0,1]$.
It is easy to prove that $H_0(-2,t)=-2$ for every $t\in[0,1]$: $H_0(-2,1)=-2$, $\varphi(H_0(-2,t))\le \varphi(-2)$ for every $t\in [0,1]$, and $-2$ is a strict local minimum point for $\varphi$.
Analogously, $H_0(2,t)=2$ for every $t\in[0,1]$. It follows that $g_0\circ f_0(-2)=-2$ and $g_0\circ f_0(2)=2$.

Now, we observe that $f_0(\{-2,2\})\subseteq\{-2,2\}$, because $f_0$ is a $0$-map and $-2,2$ are the only points where $\psi$ takes a value that is not strictly greater than $\varphi(-2)=\varphi(2)=-2$.
Analogously, $g_0(\{-2,2\})\subseteq\{-2,2\}$. By possibly composing $f_0$ with the reflection $x\mapsto -x$, we can assume that $f_0(-2)=-2$. From the equality $g_0\circ f_0(-2)=-2$, it follows that $g(-2)=-2$.
We can now prove that $f_0([-2,2])\subseteq [-2,-1]$. Indeed, since $f_0(-2)=-2$, if $f_0([-2,2])$ contained a point $\bar x>-1$, it should also contain an infinite number of points $x$ where $\psi$ takes the value $0$. This contradicts the assumption that $f_0$ is a $0$-map, because $\varphi$ takes its maximum $0$ only at the point $0$.
Furthermore, $g_0([-2,-1])\subseteq [-2,0]$. Indeed, since $g_0(-2)=-2$, if $g_0([-2,-1])$ contained a point $\bar x>0$, there should exist a point $\bar x\in [-2,-1)$ such that $g(\bar x)=0$, so that $\varphi((\bar x))=0$.
This contadicts the assumption that $g_0$ is a $0$-map.
In conclusion, we should have that $g_0\circ f_0([-2,2])\subseteq [-2,0]$, thus contradicting  the fact that
$g_0\circ f_0(2)=2$.
\end{ex}

\subsection{The importance of choosing a subcategory ${\mathbf C}$ of ${\mathbf S}$}

The main motivation for considering a subcategory ${\mathbf C}$ of ${\mathbf S}$ instead of only the category ${\mathbf S}$ is to generalize the natural pseudo-distance, whose definition depends on the selection of a set of objects and a set of morphisms that may be respectively smaller than the set of all real-valued continuous functions and the set of all homeomorphisms (cf., e.g., Section 7.1 in \cite{BiFlFa08}, \cite{CaDFLa13} and \cite{FrJa16}).  The following examples show that this choice is fruitful.

The first two examples show that the use of an appropriate subcategory ${\mathbf C}$ of ${\mathbf S}$ allows us to represent the $L^\infty$ distance and the natural pseudo-distance $\dnp$ as particular cases of $\dht^{\mathbf C}$.

\begin{ex}\label{exLinfty}
For a fixed compact $X$,  consider the category ${\mathbf C}$ whose objects are given by the pairs $(X,\bfp)$ where  $\bfp:X\to \R^n$ is continuous, and such that between any two objects $(X,\bfp),(X,\bfp')\in ob({\mathbf C})$ there is at most one morphism, $\mathrm{id}_X:X\rightarrow X$, from $(X,\bfp)$ to $(X,\bfp')$, and this happens provided that $\bfp'\preceq\bfp$. If $\bfp'$ is not everywhere less than $\bfp$, then no morphism exists from $(X,\bfp)$ to $(X,\bfp')$.
By choosing this subcategory ${\mathbf C}$ of ${\mathbf S}$ we obtain that $\dht^{\mathbf C}((X,\bfp),(X,\bfp'))=\|\bfp-\bfp'\|_\infty$.
\end{ex}

\begin{ex}\label{exNP}
Let us set $n=2m$ in the definition of the category  $\mathbf S$. Take the category ${\mathbf C}$ whose objects are the objects $(X,\bfp_X)$ of $\mathbf S$ and the morphisms from an object $(X,\bfp_X)$ to another object $(Y,\bfp_Y)$ are the homeomorphisms $f:X\to Y$ such that $\bfp_Y\circ f\preceq \bfp_X$.
If $\bfs_X:X\to\R^m$ is a continuous function, then $$\dnp((X,\bfs_X),(Y,\bfs_Y))=\dht^{\mathbf C}\left(\left(X,(\bfs_X,-\bfs_X)\right),\left(Y,(\bfs_Y,-\bfs_Y)\right)\right).$$
Here the symbol $(\bfs_X,-\bfs_X)$ denotes the function $\bfp_X:X\to\R^n$ whose first $m$ components define the function $\bfs_X$, while the last $m$ components define the function $-\bfs_X$. Analogously for the symbol $(\bfs_Y,-\bfs_Y)$.

Let us prove the previous equality. In the case that $X$ and $Y$ are not homeomorphic, we have that
$\dnp\left((X,\bfs_X),(Y,\bfs_Y)\right)=\dht^{\mathbf C}\left(\left(X,(\bfs_X,-\bfs_X)\right),\left(Y,(\bfs_Y,-\bfs_Y)\right)\right)=\infty$. Therefore we can confine ourselves to assuming that $X$ and $Y$ are homeomorphic.

If there exist an $\alpha\geq 0$ and
a homeomorphism $f:X\to Y$ such that  $\|\bfs_X- \bfs_Y\circ f\|_\infty\le\alpha$, then  
$\bfs_Y\circ f\preceq \bfs_X + \alphadiag$ and $\bfs_Y\circ f\succeq \bfs_X - \alphadiag$ (i.e. $-\bfs_Y\circ f\preceq -\bfs_X + \alphadiag$), and hence $f$ is an $\alphadiag$-map in ${\mathbf C}$ from $\left(X,(\bfs_X,-\bfs_X)\right)$ to $\left(Y,(\bfs_Y,-\bfs_Y)\right)$, with $\alphadiag=(\alpha,\alpha,\ldots,\alpha)\in\R^m$. Analogously, since $\|\bfs_Y-\bfs_X\circ f^{-1}\|_\infty=\|\bfs_X- \bfs_Y\circ f\|_\infty\le\alpha$, $f^{-1}$ is an $\alphadiag$-map in ${\mathbf C}$ from  $\left(Y,(\bfs_Y,-\bfs_Y)\right)$ to $\left(X,(\bfs_X,-\bfs_X)\right)$.
Furthermore, $f^{-1}\circ f$ and $f\circ f^{-1}$ are $\veczero$-homotopic, and hence $2\alphadiag$-homotopic to $\id_X$ and $\id_Y$, respectively, by constant homotopies. Therefore, $(f,f^{-1})$ is a pair of  $\alphadiag$-homotopy equivalences for $((\bfs_X,-\bfs_X),(\bfs_Y,-\bfs_Y))$, implying that $(X,(\bfs_X,-\bfs_X))$ and $(Y,(\bfs_Y,-\bfs_Y))$ are $(\alphadiag,\alphadiag)$-homotopy equivalent. Thus, $\dnp\left((X,\bfs_X),(Y,\bfs_Y)\right)\ge\dht^{\mathbf C}\left(\left(X,(\bfs_X,-\bfs_X)\right),\left(Y,(\bfs_Y,-\bfs_Y)\right)\right)$.

If $(X,(\bfs_X,-\bfs_X))$ and $(Y,(\bfs_Y,-\bfs_Y))$ are $(\alphadiag,\alphadiag)$-homotopy equivalent with $\alpha\geq 0$, then there is a pair $(f,g)$ of $(\alphadiag,\alphadiag)$-homotopy equivalences for $((\bfs_X,-\bfs_X),(\bfs_Y,-\bfs_Y))$. The definition of $\mathbf C$ implies that $f$ and $g$ are two homeomorphisms. Moreover,
$\bfs_Y\circ f\preceq \bfs_X + \alphadiag$ and $-\bfs_Y\circ f\preceq -\bfs_X + \alphadiag$ (i.e. $\bfs_Y\circ f\succeq \bfs_X - \alphadiag$), and hence $\|\bfs_X- \bfs_Y\circ f\|_\infty\le\alpha$. 
 
Thus, $\dnp\left((X,\bfs_X),(Y,\bfs_Y)\right)\le\dht^{\mathbf C}\left(\left(X,(\bfs_X,-\bfs_X)\right),\left(Y,(\bfs_Y,-\bfs_Y)\right)\right)$.

Therefore, $\dnp\left((X,\bfs_X),(Y,\bfs_Y)\right)=\dht^{\mathbf C}\left(\left(X,(\bfs_X,-\bfs_X)\right),\left(Y,(\bfs_Y,-\bfs_Y)\right)\right)$.

In other words, by taking a suitable subcategory $\mathbf C$ of $\mathbf S$ we can read the natural pseudo-distance $\dnp$ into the persistent homotopy type distance $\dht^{\mathbf C}$.
\end{ex}

The following examples provide more insights into the different outcomes that can be obtained varying the category $\mathbf C$. In particular, Example~\ref{exstrip} shows the effect of a restriction of both $ob({\mathbf S})$
and $hom({\mathbf S})$, while Example~\ref{exS1} illustrates the effect of a restriction of $hom({\mathbf S})$.

\begin{ex}\label{exstrip}
Let us imagine to be interested in comparing gray-level colorings of narrow strips (possibly segments), represented by pairs $(R_\epsilon, \varphi)$ where $R_\epsilon$ is the rectangle $[-1,1]\times[0,\epsilon]$ for some $\epsilon\in [0,1]$ and  $\varphi:R_\epsilon\to\R$ takes values that depend only on the first coordinate (i.e. $\varphi(x,y_1)=\varphi(x,y_2)$ for every $(x,y_1),(x,y_2)\in R_\epsilon$). We will call each of these pairs a \emph{strip coloring}. Let us also assume that we wish $(i)$ to distinguish the generic strip coloring $(R_\epsilon,\varphi)$ from its horizontal reflection $(R_\eps,\hat \varphi)$ with $\hat\varphi$ defined by setting $\hat\varphi(x,y)=\varphi(-x,y)$, and $(ii)$ not distinguish  the strip colorings $(R_\epsilon, \varphi)$ and $(R_0,\varphi_{|R_0})$ on the ground that the height of the strips is not important. Then,  these requirements can be satisfied by considering the subcategory ${\mathbf C}$ of ${\mathbf S}$ whose objects are given by the
previously defined strip colorings for all $\epsilon\in [0,1]$, and whose morphisms between two strip colorings $(R_\epsilon,\varphi),(R_{\epsilon'},\varphi')\in ob({\mathbf C})$ are the continuous maps $f=(f_1,f_2):R_{\epsilon}\to R_{\epsilon'}$ such that
\begin{enumerate}
  \item $f_1$ depends only on the first coordinate;
  \item $f_1(\cdot,0)$ is a strictly increasing homeomorphism;
  \item The inequality $\varphi'\circ f(x,y)\le \varphi(x,y)$ holds for every $(x,y)\in R_\epsilon$.
\end{enumerate}
Now, let us set $\varphi(x,y)=x$. If $f=(f_1,f_2):R_\epsilon\to R_\epsilon$ is an $\alpha$-map with respect to  $(\varphi,\hat\varphi)$,
then $\hat\varphi(f(x,y))\le \varphi (x,y)+\alpha$ for every $(x,y)\in R_\epsilon$. The definition of $\varphi$ and $\hat\varphi$ implies that $-f_1(x,y)\le x+\alpha$ for every $(x,y)\in R_\epsilon$.
By setting $x=-1$, $y=0$ and observing that $f_1(-1,0)=-1$, we get $1=-f_1(-1,0)\le -1+\alpha$, i.e. $\alpha\ge 2$.
It is immediate to check that $(\id_{R_\epsilon},\id_{R_\epsilon})$  is a pair of $2$-homotopy equivalences with respect to  $(\varphi,\hat\varphi)$. Therefore $\dht^{\mathbf C}((R_{\epsilon},\varphi),(R_{\epsilon},\hat \varphi))=2$.
Moreover,
$\dht^{\mathbf C}((R_{\epsilon},\varphi),(R_0,\varphi_{|R_0}))=0$, so that $\dht^{\mathbf C}$ satisfies the two  required properties $(i)$ and $(ii)$.
It is also interesting to observe that $\dnp((R_{\epsilon},\varphi),(R_0,\varphi_{|R_0}))=\infty$
(because $R_{\epsilon}$ and $R_0$ are not homeomorphic) and $\dht^{\mathbf S}((R_{\epsilon},\varphi),(R_\epsilon,\hat \varphi))=0$ (because the homeomorphism $f(x,y)=(-x,y)$ is a $0$-homotopy equivalence with respect to the pair $(\varphi,\hat\varphi)$),
so that neither $\dnp$ nor $\dht^{\mathbf S}$ satisfy the two  required properties $(i)$ and $(ii)$.
\end{ex}

\begin{ex}\label{exS1}
Let ${\mathbf C}$ be the category whose objects are pairs $(X,\varphi)$ with  $X=[-1,1]$ and  $\varphi:X\to \R$ a continuous function, and whose morphisms between two objects $(X,\varphi),(X,\varphi')$ are the non-decreasing continuous maps $f:X\to X$.
Let us take the two functions $\varphi,\hat\varphi:X\to \R$ defined by setting $\varphi(x)=x$ and $\hat \varphi(x)=-x$. In the category ${\mathbf C}$,  the map $f(x)\equiv 1$ is a $0$-map with respect to the pair $(\varphi,\hat\varphi)$, and the map $g(x)\equiv -1$ is a $0$-map with respect to the pair $(\hat \varphi,\varphi)$. The function $H:[-1,1]\times [0,1]\to \R$, $H(x,t)=t(x+1)-1$ is a $0$-homotopy between $g\circ f$ and $\id_X$ with respect to the pair $(\varphi,\varphi)$, and the function $H':[-1,1]\times [0,1]\to\R$, $H'(x,t)=t(x-1)+1$ is a $0$-homotopy between $f\circ g$ and $\id_X$ with respect to the pair $(\hat\varphi,\hat\varphi)$.
As a consequence, $\dht^{\mathbf C}((X,\varphi),(X,\hat\varphi))=0$. In other words, $\dht^{\mathbf C}$ cannot distinguish $\varphi$ from $\hat\varphi$.
However, if we maintain the same objects and restrict the set of morphisms to the set of all increasing homeomorphisms from $X$ to $X$, we obtain another subcategory ${\mathbf C}'$ of ${\mathbf S}$ such that
$\dht^{{\mathbf C}'}((X,\varphi),(X,\hat\varphi))>0$. Therefore, different choices of the subcategory ${\mathbf C}$ of ${\mathbf S}$ can produce different pseudo-metrics in our model.
 \end{ex}

\subsection{$\dht$ is the same in the topological, PL, and smooth categories}

In this section we prove that $\dht$ is the same in the topological, PL, and smooth categories.

\begin{prop}\label{prop:pl}
Let $\mathbf C$ be the subcategory of $\mathbf S$ such that: the objects  of $\mathbf{C}$ are all the pairs $(X,\p_X)$ where $X$ is a compact polyhedron, and $\varphi_X:X\to\R$ is a piecewise linear function; the morphisms of $\mathbf{C}$ from an object $(X,\varphi_X)$ to another object $(Y,\varphi_Y)$ are all the piecewise linear  maps $f: X\to Y$ such that $\p_Y\circ f\le \p_X$. If  $(X,\varphi_X)$ and $(Y,\varphi_Y)$ are two objects in $\mathbf C$, and hence in $\mathbf S$, then  $\dht^{\mathbf C}((X,\p_X),(Y,\p_Y))=\dht^{\mathbf S}((X,\p_X),(Y,\p_Y))$.
\end{prop}

\begin{proof}
The inequality $\dht^{\mathbf C}((X,\varphi_X),(Y,\varphi_Y))\ge \dht^{\mathbf S}((X,\varphi_X),(Y,\varphi_Y))$ holds because each morphism in $\mathbf{C}$ is also a morphism in $\mathbf{S}$. To see that the converse inequality,
let $\omega_X$ and $\omega_Y$ be the moduli of continuity of $\p_X$ and $\p_Y$, respectively. Let $(f,g)$ be a pair of $\alpha$-homotopy equivalences in ${\mathbf S}$ with respect to $(\p_X,\p_Y)$. Let $K,L$ be simplicial complexes such that $X=|K|$ and $Y=|L|$. By the simplicial approximation theorem (cf., e.g., \cite{FrPi90}), there exists $\eps>0$ and  $K^\eps$ and $L^\eps$ subdivisions of $K$ and $L$ with $mesh(K^\eps),mesh(L^\eps)<\eps$, respectively, and a PL map $ f^\eps:X\to Y$  such that, for any $x\in X$, each simplex of $L^\eps$ containing $f(x)$ contains also $ f^\eps(x)$. Since $f^\eps$ is homotopic to $f$ via $F(x,t)=(1-t)f^\eps(x)+tf(x)$, and $mesh(L^\eps)<\eps$, it holds that $ f^\eps$ is $\omega_Y(\eps)$-homotopic to $f$. Analogously, there exists a simplicial approximation $g^\eps$ of $g$ that is $\omega_X(\eps)$-homotopic to $g$  via $G(x,t)=(1-t)g^\eps(x)+tg(x)$. Notice that $f^\eps$ is an $(\alpha+\omega_Y(\eps))$-map with respect to $(\p_X,\p_Y)$ and that $g^\eps$ is an $(\alpha+\omega_X(\eps))$-map with respect to $(\p_Y,\p_X)$.  Thus, $S:X\times I\to X$ defined by $S(x,t)=G(F(x,t),t))$ is a $(\omega_X(\eps)+\omega_Y(\eps))$-homotopy between $g^\eps\circ f^\eps$ and $g\circ f$. Because $(f,g)$ is a pair of $\alpha$-homotopy equivalences with respect to $(\p_X,\p_Y)$, $g\circ f$ is $2\alpha$-homotopic to $\id_X$ with respect to $(\p_X,\p_X)$. Thus, there is a  $(\omega_X(\eps)+\omega_Y(\eps)+2\alpha)$-homotopy between $g^\eps\circ f^\eps$ and $\id_X$ with respect to $(\p_X,\p_X)$. Because any homotopy between continuous mappings can likewise be approximated by a combinatorial version, possibly further subdividing $K$ and $L$, we can approximate the $(\omega_X(\eps)+\omega_Y(\eps)+2\alpha)$-homotopy between $g^\eps\circ f^\eps$ and $\id_X$ by a $(2\omega_X(\eps)+\omega_Y(\eps)+2\alpha)$-homotopy between $g^\eps\circ f^\eps$ and $\id_X$ that is PL at each instant. Analogously, there is a $(\omega_X(\eps)+2\omega_Y(\eps)+2\alpha)$-homotopy between $f^\eps\circ g^\eps$ and $\id_Y$ with respect to $(\p_Y,\p_Y)$ that is PL at each instant. Hence, $(f^\eps,g^\eps)$ is a pair of  $(\omega_X(\eps)+\omega_Y(\eps)+\alpha))$-homotopy equivalences with respect to $(\p_X,\p_Y)$. As $\eps$ tends to 0, $\omega_X(\eps)$ and $\omega_Y(\eps)$ tend to 0. Hence, the claim.
\end{proof}

\begin{prop}
Let $\mathbf D$ be the subcategory of $\mathbf S$ such that: the objects  of $\mathbf{D}$ are all the pairs $(M,\p_M)$ where $M$ is a smooth connected compact manifold, and $\varphi_M:M\to\R$ is a smooth function; the morphisms of $\mathbf{D}$ from an object $(M,\varphi_M)$ to another object $(N,\varphi_N)$ are all the smooth  maps $f: M\to N$ such that $\p_N\circ f\le \p_M$. If  $(M,\varphi_M)$ and $(N,\varphi_N)$ are two objects in $\mathbf D$, and hence in $\mathbf S$, then  $\dht^{\mathbf D}((M,\p_M),(N,\p_N))=\dht^{\mathbf S}((M,\p_M),(N,\p_N))$.
\end{prop}

\begin{proof}
It may be proved in much the same way as Proposition~\ref{prop:pl}, using the fact that any continuous map $f:M\to N$ with $M$ and $N$ manifolds can be approximated by $C^\infty$-maps homotopic to $f$ (cf. \cite[Ch. 5, Lemma 1.5]{Hirsch}).
\end{proof}

\section{Stability of persistent homology with respect to $\dht$}\label{sec:stability-main}
In this section we establish some connections between the distance
$\dht$ and persistent homology, in particular we lift the Stability
Theorem of Persistence~\ref{thm:stab} via $\dht.$

\subsection{Preliminaries}\label{sec:stability}

\subsubsection{Overview of persistence diagrams and the bottleneck distance}\label{sec:overview}

In persistent homology, for each $k=0,1,2,\ldots$ one seeks to summarize the topological information contained in the sequence of sublevel sets $\p_X^{-1}((-\infty, t])$ into a multiset $D_k(\varphi_X)$ of points of the extended plane called a \emph{persistence diagram}: a birth–death pair $(b,d)$ corresponding to a homological feature in degree $k$ gives rise to a point $p=(b,d)$ in $D_k(\varphi_X)$; points are  taken with multiplicity in order to take into account the presence of multiple features appearing and disappearing at the same sublevels.

There is a natural notion of distance, called the {\em bottleneck distance} $\db$, that makes the set of all persistence diagrams into a metric space.
The bottleneck distance between two persistence diagrams $D_1,D_2$ is
$$\db(D_1,D_2 ) = \inf_M\max\left \{\sup_{(p,q)\in M }\|p-q\|_\infty,\sup_{\begin{array}{c}{\scriptstyle  s\in D_1\coprod D_2}\\
{\scriptstyle s\notin  M(D_1)\coprod M^{-1}(D_2)}\end{array}}\left|\frac{s_x-s_y}{2}\right|\right\},$$
where $M$ varies among all the binary relations between $D_1$ and $D_2$ that are both right- and left-unique, i.e. partial matchings between $D_1$ and $D_2$.

The following result, which for fixed $X$ expresses the continuity of the assignment $\varphi\mapsto D_k(\varphi)$, is standard:

\thmstab*

In the above statement, tameness refers to a certain regularity condition singling out functions $\varphi$ for which the homology groups $H_k(\varphi^{-1}((-\infty,a]))$ are finite dimensional for all $a\in\R$, and in addition, the maps induced at homology level by the inclusions $\varphi^{-1}((-\infty,a-\varepsilon])\hookrightarrow \varphi^{-1}((-\infty,a+\varepsilon])$ fail to be isomorphisms for  $\varepsilon>0$ small only at finitely many points.

One of the salient features of the above result is that it assumes the underlying space $X$ to be fixed. Using our construction of the homotopy type distance we lift this result into a statement that applies to any pair $(X,\varphi_X)$ and $(Y,\varphi_Y)$ in the category $\mathbf{C}$ satisfying minimal tameness conditions.

\subsubsection{Persistence modules and interleavings}

More recently,  persistent homology has been revisited in terms of persistence modules and interleavings. The main references here are  \cite{interleaving} and \cite{interleaving-brief}.

A persistent module (over $\R$) is by definition a directed sequence of vector spaces connected by linear maps $\{V_\delta\stackrel{v_{\delta,\delta'}}{\longrightarrow}V_{\delta'}\}_{\delta\leq \delta'}$  such that $v_{\delta,\delta}=\mathrm{id}$ for all $\delta\in\R$ and $v_{\delta', \delta''}\circ v_{\delta, \delta'} = v_{\delta, \delta''}$ for all $\delta\leq \delta'\leq \delta''$. It is said to be $q$-tame if the linear maps $v_{\delta,\delta'}$ with $\delta<\delta'$  have finite rank.

We now recall the notion of interleaving of persistence modules \cite[\S 3.2]{interleaving}. Given two persistent modules $\{V_\delta\stackrel{v_{\delta,\delta'}}{\longrightarrow}V_{\delta'}\}_{\delta\leq \delta' }$ and $\{W_\delta\stackrel{w_{\delta,\delta'}}{\longrightarrow}W_{\delta'}\}_{\delta\leq \delta'}$, one says that they are $\alpha\geq 0$ interleaved if for each $\delta\geq 0$ there exist maps $\phi_\delta:V_{\delta}\rightarrow W_{\delta+\alpha}$ and $\gamma_\delta:W_{\delta}\rightarrow V_{\delta+\alpha}$ such that the following four diagrams (\ref{eq:triangle-X}), (\ref{eq:triangle-Y}), (\ref{eq:quad-1}), and (\ref{eq:quad-2}) commute for all $\delta,\delta'\in\R$ with $\delta\leq \delta'$:

\begin{equation}\label{eq:triangle-X}\xymatrix{
V_{\delta}\ar[dr]_{\phi_\delta} \ar[rr]^{v_{\delta,\delta+2\alpha}} && {V}_{\delta+2\alpha}\\
 & {W}_{\delta+\alpha}\ar[ur]_{\gamma_{\delta+\alpha}}\\
}\end{equation}

\begin{equation}\label{eq:triangle-Y}\xymatrix{
 & {V}_{\delta+\alpha}\ar[dr]^{\phi_{\delta+\alpha}}\\
W_{\delta}\ar[ur]^{\gamma_\delta} \ar[rr]^{w_{\delta,\delta+2\alpha}} && {W}_{\delta+2\alpha}
}\end{equation}

\begin{equation}\label{eq:quad-1}\xymatrix{
V_{\delta} \ar[drr]_{\phi_\delta} \ar[rr]^{v_{\delta,\delta'}} && {V}_{\delta'} \ar[drr]^{\phi_{\delta'}} \\
&& {W}_{\delta+\alpha} \ar[rr]^{w_{\delta+\alpha,\delta'+\alpha}}&& {W}_{\delta'+\alpha}}
\end{equation}

\begin{equation}\label{eq:quad-2}\xymatrix{
&&V_{\delta+\alpha} \ar[rr]^{v_{\delta+\alpha,\delta'+\alpha}} && {V}_{\delta'+\alpha} \\
 {W}_{\delta} \ar[rru]^{\gamma_{\delta}}\ar[rr]^{w_{\delta,\delta'}}&& {W}_{\delta'}\ar[urr]_{\gamma_{\delta'}}}
\end{equation}


In what follows, for each non-negative integer $k$, $H_k(\cdot)$ will denote the homology functor (with field coefficients). Given a pair $(X,\varphi_X)$ and a non-negative integer $k$,  we define the associated sublevelset persistence module
$$\mathfrak{P}_k^X :=\{V_\delta\stackrel{v_{\delta,\delta'}}{\longrightarrow}V_{\delta'}\}_{\delta\leq \delta' }$$
where:

\begin{itemize}
\item for each $\delta\in\R$, $V_\delta := H_k(X^\delta)$, where $X^\delta := \varphi_X^{-1}((-\infty,\delta])$.
\item for each $\delta,\delta'\in \R$ with $\delta\leq \delta'$ the linear map $v_{\delta,\delta'}:=H_k(\iota^X_{\delta,\delta'}:X^\delta\hookrightarrow X^{\delta'})$ induced by the natural inclusion  $\iota^X_{\delta,\delta'}$ of $X^\delta$ into $X^{\delta'}$.
\end{itemize}

\subsection{Lifting stability results via $\dht$}\label{Stability of persistence diagrams}

In this section we obtain a lower bound for the homotopy type distance based on comparing persistence diagrams.

The context of the following theorem is that of  the full subcategory $\bf C$ of $\bf S$ whose objects are compact polyhedra endowed with continuous real valued functions.  Thus, the associated sublevelset persistence modules  turn out to be $q$-tame \cite{interleaving-brief}, and one can still obtain persistence diagrams for sublevelset persistence without adding any extra tameness condition on the functions \cite{CeDFal13}. Moreover, because $\bf C$ is a full subcategory of $\bf S$, $\dht^{\bf C}$ coincides with $\dht^{\bf S}$ restricted to objects of $\bf C$. Thus we can simply write $\dht$ for $\dht^{\bf C}$.

We start proving that the persistence modules of $\alpha$-homotopic pairs are $\alpha$-interleaved.

\begin{lem} \label{prop_interleaving}
Let $(X,\varphi_X)$ and $(Y,\varphi_Y)$ be two $\alpha$-homotopy equivalent pairs in $\bf C$.
Then, for every non-negative integer $k$, the persistence modules $\mathfrak{P}_k^X$ and $\mathfrak{P}_k^Y$ are $\alpha$-interleaved.
\end{lem}

\begin{proof}
Fix a non-negative integer $k$, and  write $\mathfrak{P}_k^X=\{V_\delta\stackrel{v_{\delta,\delta'}}{\longrightarrow}V_{\delta'}\}_{\delta\leq \delta' }$ and $\mathfrak{P}_k^Y=\{W_\delta\stackrel{w_{\delta,\delta'}}{\longrightarrow}W_{\delta'}\}_{\delta\leq \delta'}$ where:
\begin{itemize}
\item for each $\delta\in\R$ $V_\delta := H_k(X^\delta)$ and $W_\delta:=H_k(Y^\delta)$;
\item for each $\delta,\delta'\in \R$ with $\delta\leq \delta'$ the maps $v_{\delta,\delta'}:=H_k(\iota^X_{\delta,\delta'})$ and $w_{\delta,\delta'}:=H_k(\iota^Y_{\delta,\delta'})$.
\end{itemize}

 Let $(f,g)$ be a pair of $\alpha$-homotopy equivalences with respect to $(\varphi_X,\varphi_Y)$. For each $\delta\in\R$, let $f_\delta:=f|_{X^\delta}$ and $g_\delta:=g|_{Y^\delta}$ denote the restrictions of $f$ and $g$ to $X^\delta$ and $Y^\delta$, respectively. Notice that from the fact that $f$ and $g$ are $\alpha$-maps it follows that $\mathrm{im}(f_\delta)\subset Y^{\delta+\alpha}$ and $\mathrm{im}(g_\delta)\subset X^{\delta+\alpha}$ for each $\delta\in\R$.

In order to prove that $\mathfrak{P}_k^X$ and $\mathfrak{P}_k^Y$ are $\alpha$-interleaved, for each $\delta\in\R$ we need to provide maps $\phi_\delta:V_{\delta}\rightarrow W_{\delta+\alpha}$ and $\gamma_\delta:W_{\delta}\rightarrow V_{\delta+\alpha}$ such that the four diagrams  (\ref{eq:triangle-X}),  (\ref{eq:triangle-Y}),  (\ref{eq:quad-1}), and  (\ref{eq:quad-2}) commute for all $\delta,\delta'\in\R$ with $\delta\leq \delta'$.

In order to establish the commutativity of (\ref{eq:triangle-X}) consider the following diagram of topological spaces:
\begin{equation}\label{eq:triangle-X-topo}\xymatrix{
X^{\delta}\ar[dr]_{f_\delta} \ar[rr]^{\iota^X_{\delta,\delta+2\alpha}} && {X}^{\delta+2\alpha}\\
 & {Y}^{\delta+\alpha}\ar[ur]_{g_{\delta+\alpha}}\\
}\end{equation}

Notice that $g_{\delta+\alpha}\circ f_\delta:X^\delta\rightarrow X^{\delta+2\alpha}$ and the inclusion $\iota^X_{\delta,\delta+2\alpha}:X^\delta\hookrightarrow X^{\delta+2\alpha}$ are homotopic by hypothesis. Indeed, since a $2\alpha$-homotopy $H:X\times[0,1]\rightarrow X$ between $g\circ f$ and $\mathrm{id}_X$ exists, the restriction of $H$ to $X^\delta\times[0,1]$ has its image contained in $X^{\delta+2\alpha}$, and is therefore a proper homotopy between the maps $g_{\delta+\alpha}\circ f_\delta$ and $\iota_{\delta,\delta+2\alpha}^X$.  Thus $H_k(g_{\delta+\alpha}\circ f_\delta)=H_k(\iota^X_{\delta,\delta+2\alpha})$. Applying the homology functor to diagram (\ref{eq:triangle-X-topo}) therefore yields the commutative diagram (\ref{eq:triangle-X}). The commutativity of diagram (\ref{eq:triangle-Y}) can be established in a similar way.

In order to establish the commutativity of (\ref{eq:quad-1}) consider the following diagram of topological spaces:

\begin{equation}\label{eq:quad-1-topo}\xymatrix{
X^{\delta} \ar[drr]_{f_\delta} \ar[rr]^{\iota^X_{\delta,\delta'}} && {X}^{\delta'} \ar[drr]^{f_{\delta'}} \\
&& {Y}^{\delta+\alpha} \ar[rr]^{\iota^Y_{\delta+\alpha,\delta'+\alpha}}&& {Y}^{\delta'+\alpha}}
\end{equation}


We now verify that this diagram commutes so that the commutativity of (\ref{eq:quad-1}) follows by applying the homology functor to (\ref{eq:quad-1-topo}). Indeed, pick any $x\in X^\delta$. Then
\begin{align*}
f_{\delta'}\big(\iota^X_{\delta,\delta'}(x)\big) &=f_{\delta'}(x)\tag{Since $\iota^X_{\delta,\delta'}$ is the inclusion map}\\
&= f|_{X^{\delta'}}(x)\tag{By definition of $f_{\delta'}$}\\
&= f|_{X^{\delta}}(x)\tag{Since $x\in X^{\delta}\subseteq X^{\delta'} $}\\
&=f_{\delta}(x)\tag{Definition of $f_{\delta}$}\\
&=\iota_{\delta+\alpha,\delta'+\alpha}^Y\big(f_{\delta}(x)\big)\tag{Since $f_\delta(x)\in Y^{\delta+\alpha}\subseteq Y^{\delta'+\alpha}$.}
\end{align*}
Since $x\in X^\delta$ was arbitrary it follows that $f_{\delta'+\alpha}\circ\iota^X_{\delta,\delta'} = \iota_{\delta+\alpha,\delta'+\alpha}\circ f_{\delta+\alpha}$.

One can verify that (\ref{eq:quad-2}) commutes using a similar argument.
\end{proof}


Using  Lemma~\ref{prop_interleaving}, we now obtain the stability of persistence diagrams with respect to the persistent homotopy type distance.

\thmstabht*

\begin{proof}
Under the assumption that $X$ and $Y$ are compact polyhedra, and  $\varphi_X:X\rightarrow \R$ and $\varphi_Y:Y\rightarrow \R$ are continuous functions, the persistence modules $\mathfrak{P}_k^{X}$ and $\mathfrak{P}_k^{Y}$ are $q$-tame by \cite[Thm. 2.3]{CeDFal13}. By Lemma~\ref{prop_interleaving}, if  $(X,\varphi_X)$ and $(Y,\varphi_Y)$ are  $\alpha$-homotopy equivalent with respect to $(\varphi_X,\varphi_Y)$, then, the persistence modules $\mathfrak{P}_k^X$ and $\mathfrak{P}_k^Y$ are $\alpha$-interleaved. Thus, the claim follows from the stability theorem for $q$-tame persistence modules over $\R$ \cite{interleaving-brief}.
\end{proof}

\section{The homotopy type distance for comparing merge trees}\label{sec:mergetrees}
A merge tree is a structural descriptor used in shape analysis. For a continuous function $\p: X\to \R$ defined on a connected domain, the merge tree of $\p$ encodes how the sublevel sets $\p^{-1}((-\infty,t])$  are connected for increasing values of $t\in \R$.

Following \cite{Morozov}, it can be defined as follows. Consider the epigraph $\epi(\p)$ of $\p$, that is the space
$$\epi(\p):=\{(x,t)\in X\times \R: \p(x)\le t\},$$
and the function $\bar \p: \epi(\p)\to \R$ defined by $\bar\p(x,t)=t$. Consider the equivalence relation $\sim$ on $\epi(\p)$ defined by setting $(x,t)\sim (x',t')$ if and only if  $t=t'$ and $(x,t)$ and $(x',t')$ belong to the same connected component of $\bar\p^{-1}(t)$. The merge tree of $\p$, denoted $M_\p$, is the quotient space $\epi(\p)/\sim$. In other words, $M_\p$ is the Reeb graph \cite{Reeb} of  $\epi(\p)$ with respect to $\bar \p$.
$M_\p$ is naturally endowed with the continuous function $\hat\p:M_\p \to\R$ defined by setting $\hat\p(p):=\bar\p(x)$ for any point $x$ belonging to the equivalence class $p$.

Because $M_\p$ is the Reeb graph of $\epi(\p)$ with respect to $\bar \p$, the assumption that $X$ is a compact polyhedron and $\p$ is piecewise linear ensures that  $\epi(\p)/\sim$ is a (non-compact) polyhedron of dimension 1 \cite{DiFabioLandi}. In particular, $M_\p$ turns out to be a  tree with finitely many leaves. As such, merge trees can be compared using the persistent homotopy type distance. It is then interesting to study the persistent homotopy type distance on merge trees in relation to other distances that have been proposed for the same goal.

The manuscript \cite{Morozov} presents an interleaving distance between merge trees that is interesting because it satisfies a stability property with respect to perturbation of the function that defines the merge tree. The interleaving distance between merge trees is defined as follows. For $\eps\ge 0$, define the $\eps$-shift map $i^\eps_\p: M_\p\to M_\p$ that sends, for every $t\in\R$, a connected component of $\bar\p^{-1}(t)$ to the connected component of $\bar\p^{-1}(t+\eps)$ that contains it. In other words, $i^\eps_\p$ is the map induced by the inclusion of the sublevel sets of $\p$. Given two continuous functions $\p,\psi:X\to \R$,   and $\eps\ge 0$, consider the two $\eps$-shift maps $i^\eps_\p: M_\p\to M_\p$ and $i^\eps_\psi: M_\psi\to M_\psi$. Two continuous maps $f^\eps:M_\p\to M_\psi$ and $g^\eps:M_\psi\to M_\p$ are said to be $\eps$-compatible if the following diagrams commute
\begin{equation}
\xymatrix{
M_{\p}\ar[dr]_{f^\eps} \ar[rr]^{i_\p^{2\eps}} && {M}_{\p}\\
 & {M}_{\psi}\ar[ur]_{g^{\eps}}\\
}
\xymatrix{
& {M}_{\p}\ar[dr]^{f^{\eps}}\\
M_{\psi}\ar[ur]^{g^\eps} \ar[rr]_{i_\psi^{2\eps}} && {M}_{\psi}\\
}\end{equation}
and, moreover, $\hat \psi\circ f^{\eps}=\hat \p+\eps$ and $\hat \p\circ g^{\eps}=\hat \psi+\eps$.
In this setting, the interleaving distance, $d_I(M_\p , M_\psi)$ between two merge trees
is the greatest lower bound on $\eps$ for which there are $\eps$-compatible
maps \cite{Morozov}:
$$\di(M_\p , M_\psi) = \inf\{\eps\ge 0 : \mbox{$\exists$ $\eps$-compatible maps $f^\eps: M_\p \to M_\psi$, $g^\eps: M_\psi \to M_\p$ }\}.$$

We prove that the persistent homotopy distance on merge trees coincides with such interleaving distance.

\begin{prop}\label{prop:dht-di-mt}
For every pair of piecewise linear functions $\p,\psi : X \to 	\R$ defined on a compact connected polyhedron $X$, it holds that
$$\di(M_\p,M_\psi ) = \dht((M_\p,\hat\p), (M_\psi,\hat\psi )).$$
\end{prop}

\begin{figure}[htbp]
\begin{center}
\includegraphics[width=15cm]{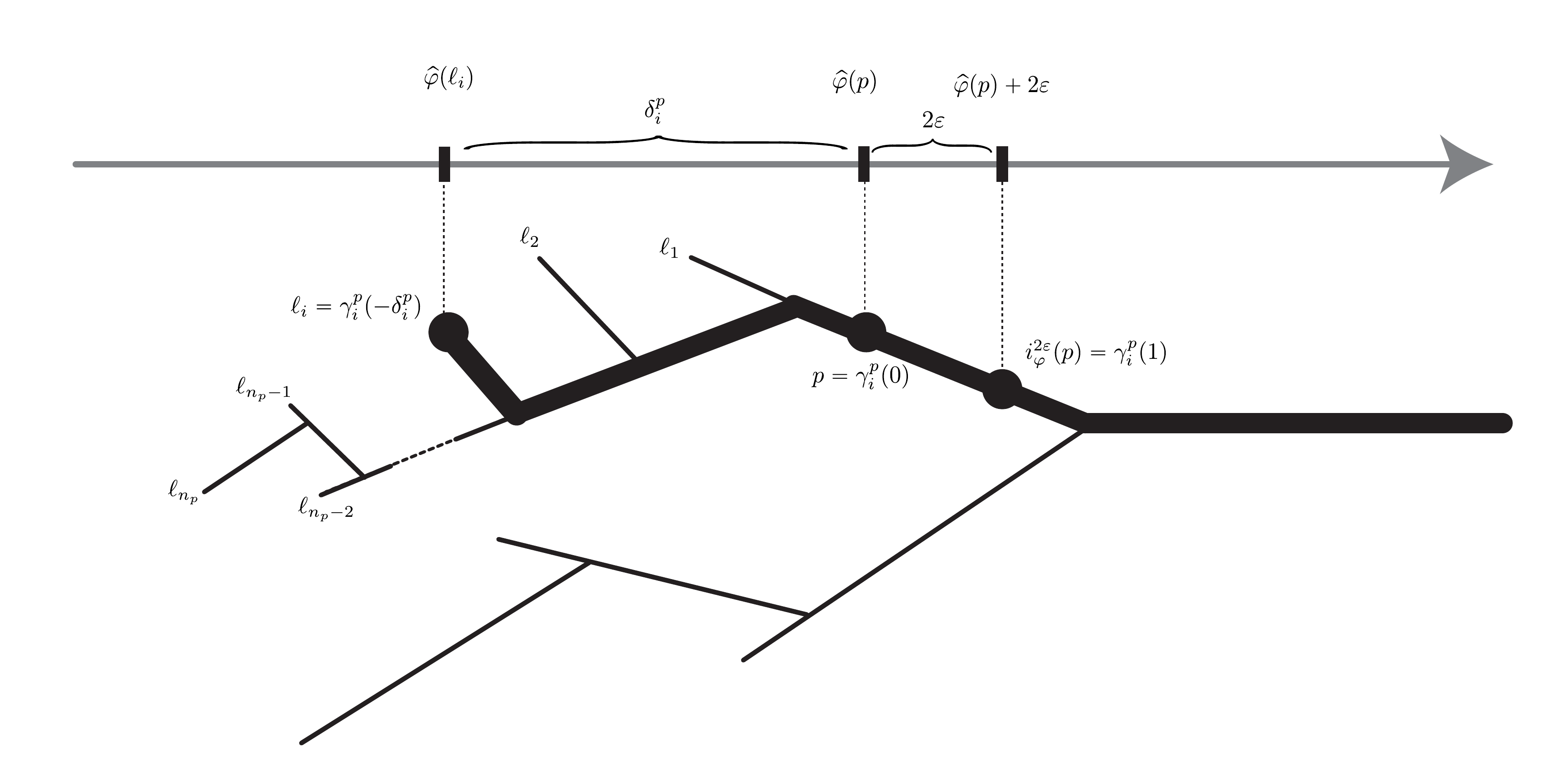}
\end{center}
\caption{Elements of the proof of Proposition~\ref{prop:dht-di-mt}.}
\label{fig:merge-tree}
\end{figure}

\begin{proof}
In order to see that $\dht\le \di$, we will prove that, given $\eps\ge 0$, every pair of $\eps$-compatible maps between $M_\p$ and $M_\psi$  constitutes a pair of $\eps$-homotopy equivalences. Pick $\eps\ge 0$ and let $i^\eps_\p: M_\p\to M_\p$ and $i^\eps_\psi: M_\psi\to M_\psi$ be $\eps$-shifts on $M_\p$ and $M_\psi$, respectively. Assume that $f^\eps:M_\p\to M_\psi$ and $g^\eps:M_\psi\to M_\p$ are  $\eps$-compatible. Because $\hat\psi\circ f^{\eps}=\hat\p+\eps$ and $\hat\p\circ g^{\eps}=\hat\psi+\eps$, it follows that $f^\eps$ and $g^\eps$ are $\eps$-maps. Let us see that $(f^\eps,g^\eps)$ is a pair of $\eps$-homotopy equivalences  with respect to $(\hat\p,\hat\psi)$. To this end, we need to construct a $2\eps$-homotopy with respect to $(\hat\p,\hat\p)$ (resp. $(\hat\psi,\hat\psi)$) between the identity of $M_\p$ (resp. $M_\psi$) and $g^\eps\circ f^\eps$  (resp. $f^\eps\circ g^\eps$). However, from the assumption that $f^\eps$ and $g^\eps$ are $\eps$-compatible, we get $i^{2\eps}_\p=g^\eps\circ f^\eps$ and $i^{2\eps}_\psi= f^\eps\circ g^\eps$. Thus, equivalently,  we need to construct a $2\eps$-homotopy $H$ (resp. $K$) with respect to $(\hat \p,\hat\p)$ (resp. $(\hat\psi,\hat\psi)$) between the identity on $M_\p$ (resp. $M_\psi$) and the $2\eps$-shift $i^{2\eps}_\p$ (resp. $i^{2\eps}_\psi$). We show how to construct $H$. Recall that $M_\p$ is a tree, and that for every $p\in M_\p$ there are finitely many leaves $\ell_1,\ldots \ell_{n_p}\in M_\p$ which can be connected to $p$ by a simple path increasing under $\hat\p$.  For $i=1,\ldots, n_p$, set $\delta_i^p:=(\hat\p(p)-\hat\p(\ell_i))/2\eps\in\R$, and let  $\gamma_i^p:[-\delta_i^p,+\infty)\to M_\p$ be the path increasing under $\hat\p$ with starting point $\ell_i$ (i.e. $\hat\p(-\delta_i^p)=\ell_i)$), passing through $p$, and parametrized by   $\hat \p\circ \gamma_i^p(t)=\hat \p(p)+2\eps t$ with $t\in [-\delta_i^p,+\infty)$. Note that, as parameter $t$ increases,  the paths $\gamma_i^p$, with $i=1,\ldots, n_p$, can only merge and not branch, and that $\gamma_i^p(0)=p$ for every $i$. See Figure~\ref{fig:merge-tree} for notation and an illustration of intervening elements.

Setting $H:M_\p\times [0,1]\to M_\p$, $(p,s)\mapsto \gamma_i^p(s)$ gives a  well-defined set map because it does not depend on the choice of the leaf $\ell_i$. Notice that $H(p,0)=p$,  $H(p,1)=i^{2\eps}_\p(p)$, and $\hat\p(H(p,s))=\hat\p(p)+2\eps s$ for all $p\in M_\p$ and $s\in [0,1]$. 

  It remains to show that $H$ is continuous. In the case when $\eps=0$,  $H(p,s)=p$ for every $(p,s)\in M_\p\times [0,1]$, and so $H$ is the identity map of $M_\p$ for every $s\in [0,1]$, implying the continuity of $H$ in that case. Now we pick $\eps>0$. Let  $q_0$ belong to the image of $H$, that is $q_0=H(p_0,s_0)$ for some $(p_0,s_0)\in M_\p\times [0,1]$, and let  $V$ be an open set in $M_\p$ containing $q_0$.
 For every $\eta>0$ and  $q\in M_\p$, let $U(q,\eta)$ be the connected component of $\hat \p^{-1}((\hat\p(q)-\eta,\hat\p(q)+\eta))$ containing $q$.
 Let $\eta>0$ be  small enough such that $U(q_0,\eta)$ is contained in $V$.  Set  $J(s_0)=(s_0-\frac{\eta}{4\eps},s_0+\frac{\eta}{4\eps})\cap [0,1]$. Note that $U(p_0,\eta/2)\times J(s_0)$ is an open neighborhood of $(p_0,s_0)$ in $M_\p\times [0,1]$.  By construction, for every $(p,s)\in U(p_0,\eta/2)\times J(s_0)$, it holds that $\hat \p(H(p,s))=\hat\p(\gamma_i^{p}(s))=\hat \p(p)+2\eps s$. Moreover, $(p,s)\in U(p_0,\eta/2)\times J(s_0)$ implies that $\hat\p(p_0)-\eta/2< \hat \p(p)< \hat\p(p_0)+\eta/2$ and $s_0-\frac{\eta}{4\eps}< s< s_0+\frac{\eta}{4\eps}$, so that
 $$\hat\p(p_0)-\frac{\eta}{2} +2\eps \left(s_0-\frac{ \eta}{4\eps}\right)< \hat \p(H(p,s))< \hat\p(p_0)+\frac{\eta}{2} +2\eps \left(s_0+\frac{ \eta}{4\eps}\right).$$
 Because $q_0=H(p_0,s_0)$ implies that $\hat \p(q_0)=\hat \p(p_0)+2\eps s_0$, we deduce that
 $$\hat\p(q_0)-\eta< \hat \p(H(p,s))< \hat\p(q_0)+\eta,$$ yielding $H(U(p_0,\eta/2)\times J(s_0))\subseteq
 U(q_0,\eta)
 \subseteq V$. 
 Indeed, we observe that $H(p,s)$ belongs to the connected component of $\hat \p^{-1}((\hat\p(q_0)-\eta,\hat\p(q_0)+\eta))$ containing $q_0$.
 In other words,  $H$ is continuous.

So, $H$ is a $2\eps$-homotopy between the identity on $M_\p$ and $i^{2\eps}_\p$.  In a similar way, we can build a $2\eps$-homotopy $K$ between the  identity on $M_\psi$ and $i^{2\eps}_\psi$. So, we have finally proved that $(f^\eps,g^\eps)$ is a pair of $\eps$-homotopy equivalences with respect to $(\hat\varphi,\hat\psi)$.

Conversely, in order to see that $\di\le \dht$, let $f:M_\p\to M_\psi$ and $g:M_\psi\to M_\p$  be such that $(f,g)$ is a pair of $\eps$-homotopy equivalences  with respect to $(\hat\varphi,\hat\psi)$ for some $\eps\ge 0$. For every $p\in M_\p$, let $$f^\eps(p):=i_\psi^{\hat \p(p)+\eps-\hat \psi(f(p))}\circ f(p).$$ Analogously, for every $q\in M_\psi$, let $$g^\eps(q):=i_\p^{\hat \psi(q)+\eps-\hat \p(g(q))}\circ g(q).$$
 In plain words, $f^\eps$ shifts the point $f(p)$ of $M_\psi$  in the direction of increasing function values until it reaches the level $\hat \p(p)+\eps$, and similarly $g^\eps$ shifts the point $g(q)$  of $M_\p$ in the direction of increasing function values until it reaches the level $\hat \psi(q)+\eps$.  By construction, $\hat \psi\circ f^\eps=\hat \p+\eps$, $\hat \p\circ g^\eps=\hat \psi+\eps$, $f^\eps\circ g^\eps=i_\psi^{2\eps}$ and $g^\eps\circ f^\eps=i_\p^{2\eps}$. To prove that $f^\eps$ and $g^\eps$ are $\eps$-compatible maps, it  remains to check that they are continuous, which is implied by the structure of the tree $M_\varphi$.

\end{proof}

\section{$\dht$ can be seen as an interleaving distance using categories}\label{sec:interleaving}
The goal of this section is to re-interpret some of the material contained in the previous sections in terms of interleavings. The advantage is that we obtain a unifying look at the distances we have encountered so far.

The theory of interleavings was initiated by Chazal et al. in \cite{interleaving}, further developed by Lesnick in \cite{Lesnick} to comprise also functional categories, and by Bubenik and Scott in \cite{Bubenik} for functors from the category of ordered reals,  and extended by Bubenik-deSilva-Scott \cite{BuSilSco} to the case of functors  from any preordered set.

\subsection{Functorial definition of interleaving distance}

For a  given integer $n\geq 1$ we denote by $\bf{R}^n$  the poset category with object set $\R^n$ with a morphism between $\bsy{u}$ and $\bsy{v}$ in $\R^n$ iff $\bsy{u}\preceq \bsy{v}$. 

Let $\bf O$ be an arbitrary category.  We start with the general definition of interleaving distance between functors from $\bf{R}^n$ to $\bf{O}$ following \cite{BuSilSco}. 


\begin{defn}\label{shiftfunctor}
Let $T: \bf{R}^n\to \bf{O}$ be a functor between $\bf{R}^n$ and $\bf O$ and  $\veceps \succeq \bsy{0}$. The {\em $\veceps$-shift	} of $T$ is the functor $T_\veceps: \bf{R}^n\to \bf{O}$ such that:
\begin{enumerate}
\item $T_\veceps(\bsy{u})=T(\bsy{u}+\veceps)$;
\item $T_\veceps(\bsy{u}\preceq \bsy{v})=T(\bsy{u}+\veceps\preceq v+\veceps)$.
\end{enumerate}
\end{defn}

Given two functors $T,T': \bf{R}^n\to \bf{O}$,  a natural transformation $\xi:T\Rightarrow T'$ between these functors consists of a morphism  $\xi_\bsy{u}:T(\bsy{u})\to T'(\bsy{u})$ in $\bf O$ for every $\bsy{u}\in \R^n$, such that, for every $\bsy{u}\preceq \bsy{v}\in\R^n$, the following diagram commutes:

\begin{eqnarray}\label{natural}
\xymatrix{
T(\bsy{u}) \ar[r]^{\xi_\bsy{u}}\ar[d]_{T(\bsy{u}\preceq \bsy{v})} &T'(\bsy{u})\ar[d]^{T'(\bsy{u}\preceq \bsy{v})} \\
T(\bsy{v})\ar[r]_{\xi_\bsy{v}}& T'(\bsy{v})}
\end{eqnarray}

Given $\varepsilon\in\R$ by $\epsdiag \in\R^n$ we denote the vector with all components equal to $\varepsilon$.

\begin{defn}\label{defn:interleaving}
Given two functors $T,T': \bf{R}^n\to \bf{O}$, and a real number $\eps\ge 0$, $T$ and $T'$ are said to be {\em $\eps$-interleaved} if there exist  natural transformations $\xi:T\Rightarrow T'_\epsdiag$ and $\eta: T'\Rightarrow T_\epsdiag$ such that, for every $\bsy{u}\in \R^n$,
$\eta_{\bsy{u}+\epsdiag}\circ \xi_\bsy{u}=T(\bsy{u}\preceq \bsy{u}+2\epsdiag)$ and  $\xi_{\bsy{u}+\epsdiag}\circ \eta_\bsy{u}=T'(\bsy{u}\preceq \bsy{u}+2\epsdiag)$.
Moreover, in that case, the pair $(\xi,\eta)$ is called an {\em $\eps$-interleaving} between $T$ and $T'$.
\end{defn}

\begin{defn}\label{di_functorial}
Given two functors $T,T': \bf{R}^n\to \bf{O}$, the {\em interleaving distance} between $T$ and $T'$ is defined as
$$\di^{\bf O}(T,T'):=\inf \{\eps\ge 0: \mbox{$T$ and $T'$ are $\eps$-interleaved}\}$$
 whenever $T$ and $T'$ are  $\eps$-interleaved for some real number  $\eps \ge 0$, whereas $\di^{\bf O}(T,T')=\infty$, otherwise.
\end{defn}

\begin{prop}[\cite{BuSilSco}]
Let $\bf {O}$ be an arbitrary category, and let $\bf {O}^{\bf{R}}$ be the category of functors from $\bf R$ to $\bf O$. It holds that $\di^{\bf O}$ is an extended pseudo-metric on the objects of $\bf {O}^{\bf{R}}$.
\end{prop}

Next, we consider specific choices of the target category $\bf O$. When there is no risk of confusion, in order to avoid an overload of notation,  we will avoid specifying the $\mathbf{O}$ symbol in the notation $\di^\mathbf{O}$.

\subsubsection{The interleaving distance between $\R^n$-valued functions}\label{sec:filtrations-n}

Let  $\bf{O}=\bf{Top}/\R_\preceq^n$  be the  category of  continuous functions $\p_X:X\to \R^n$ as objects, and $\veczero$-maps as morphisms: a morphism from $\bfp_X$ to $\bfp_Y$ is a continuous map $f:X\to Y$ such that $\bfp_Y\circ f\preceq \bfp_X$.

\begin{prop}\label{defT-n}
Every continuous function $\bfp_X:X\to \R^n$ defines a functor $T^{\bfp_X}: \bf{R}^n\to \bf{Top}/\R_\preceq^n$   by setting
\begin{itemize}
\item for every  $\bsy{u}\in\R^n$, $T^{\bfp_X}(\bsy{u}):=\bfp_X-\bsy{u}$;
\item for every $\bsy{u}, \bsy{v}\in \R^n$ with $\bsy{u}\preceq \bsy{v}$, $T^{\bfp_X}(\bsy{u}\preceq \bsy{v}):=\id_X$.
\end{itemize}
\end{prop}

\begin{proof}
For $\bsy{u}\preceq \bsy{v}$, $\id_X:X\to X$ is a morphism between $\bfp_X-\bsy{u}$ and $\bfp_X-\bsy{v}$ because  $(\bfp_X-\bsy{v})\circ \id_X\preceq \bfp_X-\bsy{u}$. Moreover, $T^{\bfp_X}$ preserves identity and composition.
\end{proof}

Let us now see what $\veceps$-shifts look like in this setting.

\begin{prop}\label{shift-n}
For every $\veceps\succeq \bsy{0}$, and for every $\bfp_X$, the $\veceps$-shift of $T^{\bfp_X}$, $T^{\bfp_X}_\veceps:\bf{R}^n\to \bf{Top}/\R_\preceq^n$, is equal to the functor $T^{\bfp_X-\veceps}:\bf{R}^n\to \bf{Top}/\R^n_\preceq$.
\end{prop}

\begin{proof}
By Definition \ref{shiftfunctor} and the definition of $T^{\bfp_X}$ in Proposition \ref{defT-n}, the $\veceps$-shift of $T^{\bfp_X}$ is the functor $T^{\bfp_X}_\veceps:\bf{R}^n\to \bf{Top}/\R_\preceq^n$ such that
\begin{enumerate}
\item for every  $\bsy{u}\in\R^n$, $T^{\bfp_X}_\veceps(\bsy{u})=T^{\bfp_X}(\bsy{u}+\veceps)=\bfp_X-\bsy{u}-\veceps=T^{\bfp_X-\veceps}(\bsy{u})$;
\item for every $\bsy{u}, \bsy{v}\in \R^n$ with $\bsy{u}\preceq \bsy{v}$, $T^{\bfp_X}_\veceps(\bsy{u}\preceq \bsy{v})=T^{\bfp_X}(\bsy{u}+\veceps\preceq \bsy{v}+\veceps)=\id_X=T^{\bfp_X-\veceps}(\bsy{u}\preceq \bsy{v})$.
\end{enumerate}
\end{proof}

The next lemma characterizes all the natural transformations between pairs of  functors $T^{\bfp_X}$ and $T^{\bfp_Y}$.

\begin{lem}\label{globalmap-n}
Let $X$, $Y$ be topological spaces, and let $\bfp_X:X\to \R^n$, $\bfp_Y:Y\to \R^n$ be any two continuous functions. Then, every continuous map $f : X 	\to Y$ such that $\bfp_Y \circ f \preceq  \bfp_X$ induces a natural transformation $\xi^f : T^{\bfp_X}\Rightarrow T^{\bfp_Y}$ such that, for every $\bsy{u}\in\R^n$, $\xi^f(\bsy{u})$ is equal to $f$.
 Reciprocally, to every natural transformation $\xi:T^{\bfp_X}\Rightarrow T^{\bfp_Y}$
corresponds a continuous map $f : X \to Y$ such that $\bfp_Y \circ f \preceq \bfp_X$ defined by $f=\xi(\bsy{u})$ for any $\bsy{u}\in\R^n$.
\end{lem}

\begin{proof}
Any continuous map $f:X\to Y$ such that ${\bfp_Y}\circ f\preceq {\bfp_X}$ induces a natural transformation $\xi^f:T^{\bfp_X}\Rightarrow T^{\bfp_Y}$ defined as follows: for every $\bsy{u}\in \R^n$, $\xi^f(\bsy{u}):=f$. Indeed, for every $\bsy{u}\in \R^n$, $f$ is a morphism between ${\bfp_X}-\bsy{u}$ and ${\bfp_Y}-\bsy{u}$ because  ${\bfp_Y}\circ f\preceq {\bfp_X}$ implies $({\bfp_Y}-\bsy{u})\circ f\preceq {\bfp_X}-\bsy{u}$,  proving that $\xi^f(\bsy{u}): T^{\bfp_X}(\bsy{u})\to T^{\bfp_Y}(\bsy{u})$. Moreover, $f\circ \id_X=\id_Y\circ f$, proving that $\xi^f(\bsy{u})\circ T^{\bfp_X}(\bsy{u}\preceq \bsy{v})=T^{\bfp_Y}(\bsy{u}\preceq \bsy{v})\circ \xi^f(\bsy{v})$.

Reciprocally,  assume that there is a family  $\xi=\{\xi(\bsy{u}): T^{\bfp_X}(\bsy{u})\to T^{\bfp_Y}(\bsy{u}) \}_{\bsy{u}\in \R^n}$ of continuous maps such that diagram (\ref{natural}) commutes for every $\bsy{u}\preceq \bsy{v}\in \R^n$: $$\xi(\bsy{u})\circ T^{\bfp_X}(\bsy{u}\preceq \bsy{v})=T^{\bfp_Y}(\bsy{u}\preceq \bsy{v})\circ \xi(\bsy{v}).$$ 
 Equivalently, for $\bsy{u}\preceq \bsy{v}$, $\xi(\bsy{v})\circ \id_X=\id_Y\circ \xi(\bsy{u})$. Hence, $\xi(\bsy{u})=\xi(\bsy{v})$ for every $\bsy{u}\preceq \bsy{v}$. Hence, it is sufficient to define $f:=\xi(\bsy{u})$ for one and hence all $\bsy{u}\in\R^n$.
\end{proof}

We can now show that the interleaving distance between functors $T^{\bfp_X}$ and $T^{\bfp_Y}$ in $\bf{Top}/\R_\preceq^n$ coincides with the natural pseudo-distance.


\begin{prop}\label{NP_interleaving}
Let $X$ and $Y$ be topological spaces. Denoting by $\dnp$ the natural pseudo-distance, for every pair of continuous functions ${\bfp_X}:X\to \R^n$ and ${\bfp_Y}:Y\to \R^n$, it holds that
$$\di(T^{\bfp_X},T^{\bfp_Y})=\dnp((X,{\bfp_X}),(Y,{\bfp_Y})).$$
\end{prop}

\begin{proof}
If $X$ and $Y$ are not homeomorphic, $\dnp=\infty$ trivially implying that $\di\le \dnp$.
 Let $f:X\to Y$ be a homeomorphism such that $\|{\bfp_X}-{\bfp_Y}\circ f\|_\infty\le \eps$ for some $\eps\ge 0$. It then holds that $({\bfp_Y}-\epsdiag)\circ f\preceq {\bfp_X}$ and $({\bfp_X}-\epsdiag)\circ f^{-1}\preceq {\bfp_Y}$. By Lemma~\ref{globalmap-n}, $f$ induces a natural transformation $\xi^f$ from $T^{\bfp_X}$ to $T^{({\bfp_Y}-\epsdiag)}$, as well as a natural transformation $\xi^{f^{-1}}$ from $T^{\bfp_Y}$ to $T^{({\bfp_X}-\epsdiag)}$. By Proposition~\ref{shift-n},  $T^{({\bfp_Y}-\epsdiag)}=T^{\bfp_Y}_\epsdiag$ and $T^{({\bfp_X}-\epsdiag)}=T^{\bfp_X}_\epsdiag$. Moreover, for every $\bsy{u}\in \R^n$, $\xi^f(\bsy{u})=f$, and analogously $\xi^{f^{-1}}(\bsy{u}+\epsdiag)=f^{-1}$. Thus, we get  $\xi^{f^{-1}}(\bsy{u}+\epsdiag)\circ \xi^f(\bsy{u})=f^{-1}\circ f=\id_X=T^{\bfp_X}(\bsy{u}\preceq \bsy{u}+2\epsdiag)$. Similarly, $\xi^{f}({\bsy{u}+\epsdiag})\circ \xi^{f^{-1}}(\bsy{u})=f\circ f^{-1}=\id_Y=T^{\bfp_Y}(\bsy{u}\preceq \bsy{u}+2\epsdiag)$. Hence, the pair $(\xi^f,\xi^{f^{-1}})$ is an  $\eps$-interleaving between $T^{\bfp_X}$ and $T^{\bfp_Y}$, proving that $\di\le \dnp$ also in this case.

Let us now show that $\dnp\le \di$. The claim is obvious if $\di=\infty$, so let us assume that $T^{\bfp_X}$ and $T^{\bfp_Y}$ are $\eps$-interleaved for some $\eps\ge 0$ by an $\eps$-interleaving $(\xi,\eta)$ with $\xi:T^{\bfp_X}\Rightarrow T_\epsdiag^{\bfp_Y}$ and $\eta:T^{\bfp_Y}\Rightarrow T_\epsdiag^{\bfp_X}$. By Definition \ref{defn:interleaving}, for every $\bsy{u}\in\R^n$, $\eta(\bsy{u}+\epsdiag)\circ \xi(\bsy{u})= T^{\bfp_X}(\bsy{u}\preceq \bsy{u}+2\epsdiag)$, and $\xi(\bsy{u}+\epsdiag)\circ \eta(\bsy{u})= T^{\bfp_Y}(\bsy{u}\preceq \bsy{u}+2\epsdiag)$.  Hence, by Lemma~\ref{globalmap-n} and Proposition~\ref{shift-n}, there are two continuous maps $f:X\to Y$ and $g:Y\to X$ such that, for each $\bsy{u}\in\R^n$, $\xi(\bsy{u})=f$, $\eta(\bsy{u})=g$, $({\bfp_Y}-\epsdiag)\circ f\preceq {\bfp_X}$ and $({\bfp_X}-\epsdiag)\circ g\preceq {\bfp_Y}$. Moreover, from $\eta(\bsy{u}+\epsdiag)\circ \xi(\bsy{u})=T^{\bfp_X}(\bsy{u}\preceq \bsy{u}+2\epsdiag)$ and $\xi(\bsy{u}+\epsdiag)\circ \eta(\bsy{u})= T^{\bfp_Y}(\bsy{u}\preceq \bsy{u}+2\epsdiag)$, it follows that $g\circ f=\id_X$ and $f\circ g=\id_Y$. Hence, $g=f^{-1}$ and $\|{\bfp_X}-{\bfp_Y}\circ f\|_\infty\le \eps$.
\end{proof}


\subsubsection{The interleaving distance between functions up to homotopy}
 We now consider $\bf O= \bf{hTop}/\R^n_\preceq$, the  category  with continuous functions $\bfp_X:X\to \R^n$ as objects, and  $\veczero$-homotopy classes of $\veczero$-maps between $X$ and $Y$ as morphisms: for two objects $\bfp_X$,  $\bfp_Y$ in $\bf{hTop}/\R^n_\preceq$,  a morphism  from $\bfp_X$ to $\bfp_Y$ is  the $\veczero$-homotopy class with respect to $(\bfp_X,\bfp_Y)$ of  a $\veczero$-map $f:X\to Y$, that is, a continuous map between $X$ and $Y$ such that  $\bfp_Y\circ f\preceq \bfp_X$, and is denoted by $[f]_{(\bfp_X,\bfp_Y)}$.  In $\bf{hTop}/\R^n_\preceq$, the composition of morphisms  is defined as the $\veczero$-homotopy class of the composition of  $\veczero$-maps: for $\bfp_X$,  $\bfp_Y$, $\bfp_Z$ in $\bf{hTop}/\R^n_\preceq$, and $f:X\to Y$, $g:Y\to Z$ being $\veczero$-maps with respect to $(\bfp_X,\bfp_Y)$ and $(\bfp_Y,\bfp_Z)$, respectively, we set 
$$[g]_{(\bfp_Y,\bfp_Z)}\circ [f]_{(\bfp_X,\bfp_Y)}=[g\circ f]_{(\bfp_X,\bfp_Z)}.$$ 
This is well defined because the composition does not depend on the representatives, and $\bfp_Z\circ g\le \bfp_Y$, $\bfp_Y\circ f\preceq \bfp_X$ imply $\bfp_Z\circ g\circ f\preceq \bfp_X$. Composition is associative and, for any object $(X,\bfp_X)$,  the $\veczero$-homotopy class of the identity $\id_X$ with respect to $(\bfp_X, \bfp_X)$ is a morphism from $(X,\bfp_X)$ to $(X,\bfp_X)$ in $\bf{hTop}/\R^n_\preceq$.

\begin{prop}\label{defhT}
Every continuous function $\bfp_X:X\to \R^n$ defines a functor $hT^{\bfp_X}: \bf{R}\to \bf{hTop}/\R^n_\preceq$   by setting
\begin{itemize}
\item for every  $\bsy{u}\in\R^n$, $hT^{\bfp_X}(\bsy{u}):=\bfp_X-\bsy{u}$;
\item for every $\bsy{u}, \bsy{v}\in \R^n$ with $\bsy{u}\preceq \bsy{v}$, $hT^{\bfp_X}(\bsy{u}\preceq \bsy{v}):=[\id_X]_{(\bfp_X-\bsy{u},\bfp_X-\bsy{v})}$.
\end{itemize}
\end{prop}

\begin{proof}
For $\bsy{u}\preceq \bsy{v}$, $[\id_X]_{(\bfp_X-\bsy{u},\bfp_X-\bsy{v})}$ is a morphism between $\bfp_X-\bsy{u}$ and $\bfp_X-\bsy{v}$ because $(\bfp_X-\bsy{v})\circ \id_X\preceq \bfp_X-\bsy{u}$. Moreover, $hT^{\bfp_X}$ preserves the identity and the composition.
\end{proof}

Let us now consider $\veceps$-shifts of functors $hT^{\bfp_X}$.

\begin{prop}\label{h-shift}
For every $\veceps\succeq \veczero$, and for every $\bfp_X$, the $\veceps$-shift of $hT^{\bfp_X}$, $hT^{\bfp_X}_\veceps:\bf{R}^n\to \bf{hTop}/\R^n_\preceq$, is equal to the functor $hT^{\bfp_X-\veceps}:\bf{R}^n\to \bf{hTop}/\R^n_\preceq$.
\end{prop}

\begin{proof}
By Definition \ref{shiftfunctor} and the definition of $hT^{\bfp_X}$ (Proposition \ref{defhT}), the $\veceps$-shift of $hT^{\bfp_X}$ is the functor $hT^{\bfp_X}_\veceps:\bf{R}^n\to \bf{hTop}/\R^n_\preceq$ such that
\begin{enumerate}
\item for every  $\bsy{u}\in\R^n$, $hT^{\bfp_X}_\veceps(\bsy{u})=hT^{\bfp_X}(\bsy{u}+\veceps)=\bfp_X-\bsy{u}-\veceps=hT^{\bfp_X-\veceps}(\bsy{u})$;
\item for every $\bsy{u}, \bsy{v}\in \R^n$ with $\bsy{u}\preceq \bsy{v}$, $hT^{\bfp_X}_\veceps(\bsy{u}\preceq \bsy{v})=hT^{\bfp_X}(\bsy{u}+\veceps\preceq \bsy{v}+\veceps)=[\id_X]_{(\bfp_X-\bsy{u}-\veceps,\bfp_X-\bsy{v}-\veceps)}= hT^{\bfp_X-\veceps}(\bsy{u}\preceq \bsy{v})$.
\end{enumerate}
\end{proof}

The next lemma describes the natural transformations between pairs of  functors $hT^{\bfp_X}$ and $hT^{\bfp_Y}$.

\begin{lem}\label{h-globalmap}
Let $X$, $Y$ be topological spaces, and let $\bfp_X:X\to \R^n$, $\bfp_Y:Y\to \R^n$ be any two continuous functions. Then, every continuous map $f : X\to Y$ such that $\bfp_Y \circ f \preceq  \bfp_X$ induces a natural transformation $h\xi^f : hT^{\bfp_X}\Rightarrow hT^{\bfp_Y}$ such that, for every $\bsy{u}\in\R^n$, $h\xi^f(\bsy{u})$ is equal to $[f]_{(\bfp_X-\bsy{u},\bfp_Y-\bsy{u})}$.
 Reciprocally, for every natural transformation $h\xi:hT^{\bfp_X}\Rightarrow hT^{\bfp_Y}$, with $h\xi(\bsy{u})=[f_\bsy{u}]_{(\bfp_X-\bsy{u},\bfp_Y-\bsy{u})}$, and for every $\vecalpha\succeq \veczero$ in $\R^n$, $f_\bsy{u}$ and $f_{\bsy{u}+\vecalpha}$ are $\vecalpha$-homotopic with respect to $(\bfp_X,\bfp_Y)$ for every $\bsy{u}\in\R^n$.  
\end{lem}

\begin{proof}
Any continuous map $f:X\to Y$ such that ${\bfp_Y}\circ f\preceq {\bfp_X}$  induces a natural transformation $h\xi^f:hT^{\bfp_X}\Rightarrow hT^{\bfp_Y}$ defined as follows: for every $\bsy{u}\in \R^n$, $h\xi^f(\bsy{u}):=[f]_{(\bfp_X-\bsy{u},\bfp_Y-\bsy{u})}$. Indeed, for every $\bsy{u}\in \R^n$, $f$ is a morphism between ${\bfp_X}-\bsy{u}$ and ${\bfp_Y}-\bsy{u}$ because  ${\bfp_Y}\circ f\preceq {\bfp_X}$ implies $({\bfp_Y}-\bsy{u})\circ f\preceq {\bfp_X}-\bsy{u}$,  proving that $h\xi^f(\bsy{u}): hT^{\bfp_X}(\bsy{u})\to hT^{\bfp_Y}(\bsy{u})$. Moreover, $[f]_{(\bfp_X-\bsy{u},\bfp_Y-\bsy{u})}\circ [\id_X]_{(\bfp_X-\bsy{u},\bfp_X-\bsy{u})}=[\id_Y]_{(\bfp_Y-\bsy{u},\bfp_Y-\bsy{u})}\circ [f]_{(\bfp_X-\bsy{u},\bfp_Y-\bsy{u})}$, proving that $h\xi^f(\bsy{u})\circ hT^{\bfp_X}(\bsy{u}\preceq \bsy{v})=hT^{\bfp_Y}(\bsy{u}\preceq \bsy{v})\circ h\xi^f(\bsy{v})$.

Reciprocally,  assume that there is a family  $h\xi=\{h\xi(\bsy{u}): hT^{\bfp_X}(\bsy{u})\to hT^{\bfp_Y}(\bsy{u}) \}_{\bsy{u}\in \R^n}$ of $\veczero$-homotopy classes of maps $f_\bsy{u}$ with respect to $(\bfp_X-\bsy{u},\bfp_Y-\bsy{u})$ such that diagram (\ref{natural}) commutes for every $\bsy{u}\le \bsy{v}\in \R^n$: $h\xi(\bsy{v})\circ hT^{\bfp_X}(\bsy{u}\preceq \bsy{v})=hT^{\bfp_Y}(\bsy{u}\preceq \bsy{v})\circ h\xi(\bsy{u})$. Take $\vecalpha=\bsy{v}-\bsy{u}$. From 
$$h\xi(\bsy{u}+\vecalpha)\circ hT^{\bfp_X}(\bsy{u}\preceq \bsy{u}+\vecalpha)=hT^{\bfp_Y}(\bsy{u}\preceq \bsy{u}+\vecalpha)\circ h\xi(\bsy{u})$$ it follows that 
$$[f_{\bsy{u}+\vecalpha}]_{(\bfp_X-\bsy{u}-\vecalpha,\bfp_Y-\bsy{u}-\vecalpha)}\circ [\id_X]_{(\bfp_X-\bsy{u},\bfp_X-\bsy{u}-\vecalpha)}=[\id_Y]_{(\bfp_Y-\bsy{u},\bfp_Y-\bsy{u}-\vecalpha)}\circ [f_\bsy{u}]_{(\bfp_X-\bsy{u},\bfp_Y-\bsy{u})}.$$ 
Equivalently, $[f_{\bsy{u}+\vecalpha}]_{(\bfp_X-\bsy{u},\bfp_Y-\bsy{u}-\vecalpha)}=[f_\bsy{u}]_{(\bfp_X-\bsy{u},\bfp_Y-\bsy{u}-\vecalpha)}$. In other words, there exists a homotopy $H:X\times I\to Y$ such that $H(\cdot , 0)=f_{\bsy{u}}$, $H(\cdot , 1)=f_{\bsy{u}+\vecalpha}$, and $\bfp_Y\circ H(\cdot, t)-\bsy{u}-\vecalpha\preceq \bfp_X-\bsy{u}$ for every $t\in I$. Hence $\bfp_Y\circ H(\cdot, t)\preceq \bfp_X+\vecalpha$ for every $t\in I$, yielding that $f_\bsy{u}$ and $f_{\bsy{u}+\vecalpha}$ are $\vecalpha$-homotopic with respect to $(\bfp_X,\bfp_Y)$. 
\end{proof}

We can now show that the interleaving distance between functors $hT^{\bfp_X}$ and $hT^{\bfp_Y}$ in $\bf{hTop}/\R^n_\preceq$ coincides with the persistent homotopy type pseudo-distance.


\begin{prop}\label{HT_interleaving}
Let $X$ and $Y$ be topological spaces. Denoting by $\dht$ the persistent homotopy type pseudo-distance, for every pair of continuous functions ${\bfp_X}:X\to \R^n$ and ${\bfp_Y}:Y\to \R^n$, it holds that
$$\di(hT^{\bfp_X},hT^{\bfp_Y})=\dht((X,{\bfp_X}),(Y,{\bfp_Y})).$$
\end{prop}

\begin{proof}
If $X$ and $Y$ are not homotopy equivalent, $\dht=\infty$ trivially implying that $\di\le \dht$.
 Let $f:X\to Y$, $g:Y\to X$ form a pair $(f,g)$ of $\veceps$-homotopy equivalences with respect to $(\bfp_X,\bfp_Y)$ for some $\veceps\succeq \veczero$. It holds that $({\bfp_Y}-\veceps)\circ f\preceq {\bfp_X}$ and $({\bfp_X}-\veceps)\circ g\preceq {\bfp_Y}$. By Lemma~\ref{h-globalmap}, $f$ induces the natural transformation $h\xi^f:hT^{\bfp_X}\Rightarrow hT^{({\bfp_Y}-\veceps)}$ with $h\xi^f(\bsy{u})=[f]_{(\bfp_X-\bsy{u},\bfp_Y-\veceps-\bsy{u})}$, and $g$ induces the natural transformation $h\xi^{g}:hT^{\bfp_Y}\Rightarrow hT^{({\bfp_X}-\veceps)}$ with $h\xi^{g}(\bsy{u})=[g]_{(\bfp_Y-\bsy{u},\bfp_X-\veceps-\bsy{u})}$. By Proposition~\ref{h-shift},  $hT^{({\bfp_Y}-\veceps)}=hT^{\bfp_Y}_\veceps$ and $hT^{({\bfp_X}-\veceps)}=hT^{\bfp_X}_\veceps$.  Thus, we get 
 
\begin{eqnarray*}
h\xi^{g}(\bsy{u}+\veceps)\circ h\xi^f(\bsy{u})&=&[g]_{(\bfp_Y-\veceps-\bsy{u},\bfp_X-2\veceps-\bsy{u})}\circ [f]_{(\bfp_X-\bsy{u},\bfp_Y-\veceps-\bsy{u})}\\
&=&[g\circ f]_{(\bfp_X-\bsy{u},\bfp_X-2\veceps-\bsy{u})}=[g\circ f]_{(\bfp_X,\bfp_X-2\veceps)}\\
&=&[\id_X]_{(\bfp_X,\bfp_X-2\veceps)}=hT^{\bfp_X}(\bsy{u}\preceq \bsy{u}+2\veceps).
\end{eqnarray*}
Similarly, $h\xi^{f}({\bsy{u}+\veceps})\circ h\xi^{g}(\bsy{u})=hT^{\bfp_Y}(\bsy{u}\preceq \bsy{u}+2\veceps)$. Hence, the pair $(h\xi^f,h\xi^{g})$ is an  $\|\veceps\|_\infty$-interleaving between $hT^{\bfp_X}$ and $hT^{\bfp_Y}$, proving that $\di\le \dht$ also in this case.

Let us now show that $\dht\le \di$. The claim is obvious if $\di=\infty$, so let us assume that $hT^{\bfp_X}$ and $hT^{\bfp_Y}$ are $\eps$-interleaved for some $\eps\geq 0$ by an $\eps$-interleaving $(h\xi,h\zeta)$.   By  Proposition~\ref{h-shift}, there is a continuous $\veczero$-map $f_\bsy{u}:X\to Y$ with respect to $(\bfp_X-\bsy{u},\bfp_Y-\bsy{u}-\epsdiag)$ and a continuous $\veczero$-map $g_{\bsy{u}+\epsdiag}:Y\to X$ with respect to $(\bfp_Y-\bsy{u}-\epsdiag,\bfp_X-\bsy{u}-2\epsdiag)$ such that  $ [f_\bsy{u}]_{(\bfp_X-\bsy{u},\bfp_Y-\bsy{u}-\epsdiag)}=h\xi(\bsy{u})$ and  $[g_{\bsy{u}+\epsdiag}]_{(\bfp_Y-\bsy{u}-\epsdiag,\bfp_X-\bsy{u}-2\epsdiag)}=h\zeta(\bsy{u}+2\epsdiag)$ for any fixed $\bsy{u}$ arbitrarily chosen. Note that $f_\bsy{u}$ and $g_{\bsy{u}+\epsdiag}$ are also  $\epsdiag$-maps with respect to $(\bfp_X,\bfp_Y)$ and $(\bfp_Y,\bfp_X)$, respectively.  By Definition \ref{defn:interleaving}, for every $\bsy{u}\in\R^n$, $h\zeta(\bsy{u}+\epsdiag)\circ h\xi(\bsy{u})= hT^{\bfp_X}(\bsy{u}\le u+2\epsdiag)$, and $h\xi(\bsy{u}+2\epsdiag)\circ h\zeta(\bsy{u}+\epsdiag)= hT^{\bfp_Y}(\bsy{u}+\epsdiag\preceq \bsy{u}+3\epsdiag)$. Hence,  from $h\zeta(\bsy{u}+\epsdiag)\circ h\xi(\bsy{u})= hT^{\bfp_X}(\bsy{u}\preceq \bsy{u}+2\epsdiag)$ we deduce that $g_{\bsy{u}+\epsdiag}\circ f_\bsy{u}$ is $\veczero$-homotopic to $\id_X$ with respect to $(\bfp_X-\bsy{u},\bfp_X-\bsy{u}-2\epsdiag)$, that is $g_{\bsy{u}+\epsdiag}\circ f_\bsy{u}$ is $2\epsdiag$-homotopic to $\id_X$ with respect to $(\bfp_X,\bfp_X)$. Analogously, from $h\xi(\bsy{u}+\epsdiag)\circ h\zeta(\bsy{u})= hT^{\bfp_Y}(\bsy{u}\preceq \bsy{u}+2\epsdiag)$ we deduce that $f_{\bsy{u}+\epsdiag}\circ g_\bsy{u}$ is $\veczero$-homotopic to $\id_Y$ with respect to $(\bfp_Y-\bsy{u},\bfp_Y-\bsy{u}-2\epsdiag)$, that is $f_{\bsy{u}+\epsdiag}\circ g_\bsy{u}$ is $2\epsdiag$-homotopic to $\id_Y$ with respect to $(\bfp_Y,\bfp_Y)$. On the other hand, by Lemma \ref{h-globalmap}, $g_\bsy{u}$ is $\epsdiag$-homotopic to $g_{\bsy{u}+\epsdiag}$ with respect to $(\bfp_Y, \bfp_X)$, and $f_{\bsy{u}+\epsdiag}$ is $\epsdiag$-homotopic to $f_{\bsy{u}}$ with respect to $(\bfp_X, \bfp_Y)$. Thus, $f_{\bsy{u}+\epsdiag}\circ g_\bsy{u}$ is $2\epsdiag$-homotopic to  $f_\bsy{u}\circ g_{\bsy{u}+\veceps}$ with respect to $(\bfp_Y, \bfp_Y)$, implying that $f_\bsy{u}\circ g_{\bsy{u}+\veceps}$ is $2\epsdiag$-homotopic to $\id_Y$ with respect to $(\bfp_Y,\bfp_Y)$. In conclusion, for an arbitrarily chosen $\bsy{u}\in\R^n$, $(f_\bsy{u}, g_{\bsy{u}+\epsdiag})$ is a pair of $\epsdiag$-homotopy equivalences with respect to $(\bfp_X,\bfp_Y)$,  thus proving that $\dht((X,\bfp_X),(Y,\bfp_Y))\le \eps$.
\end{proof}

\subsubsection{The interleaving distance between persistence modules}

Let ${\bf Vect}_{\mathbb F}$ be the category of vector spaces over a fixed field ${\mathbb F}$. An $n$-dimensional persistence module can be viewed as a functor \mbox{$\mathfrak{P}:\bf{R}^n\to \bf{Vect}_{\mathbb F}$}. Replacing $\bf{O}$ with ${\bf Vect}_{\mathbb F}$ in Definition~\ref{defn:interleaving}, we obtain the interleaving distance on the category ${\bf Vect}_{\mathbb F}$ as usual.


A standard way in which one obtains $n$-dimensional persistence modules from topological spaces endowed with $\R^n$ valued functions, is by composing the sublevelset filtration functor with the homology functor: for $k$ a non-negative integer, and for $(X,\bfp_X)$ an object in  ${\bf Top}/\R^n_\preceq$, $\mathfrak{P}_k^X= H_k\circ T^{\bfp_X}$ is an object in ${\bf Vect}_{\mathbb F}^{{\bf R}^n}$. 


An immediate remark is that given $k\in\N$ and any $\eps$-interleaving $(\zeta,\eta)$ between objects $(X,\bfp_X)$ and $(Y,\bfp_Y)$ in  ${\bf Top}/\R^n_\preceq$, the pair $(H_k(\zeta),H_k(\eta))$ is an $\eps$-interleaving between the $n$-dimensional persistence modules $\mathfrak{P}_k^X$ and $\mathfrak{P}_k^Y$.

As a corollary of Propositions \ref{HT_interleaving} and \ref{NP_interleaving} we then obtain the following result linking the interleaving distances on the categories ${\bf Vect}_{\mathbb F}^{\mathbf{R}^n}$, ${\bf Top}/\R^n_\preceq$, and ${{\bf hTop}/\R^n_\preceq}$.

\begin{cor}
For every $(X,\bfp_X), (Y,\bfp_Y)$ with $X,Y$  topological spaces, every $\bfp_X:X\to \R^n$, $\bfp_Y:Y\to \R^n$  continuous functions, and every non-negative integer $k$, it holds that
$$\di(\mathfrak{P}_k^X, \mathfrak{P}_k^Y)\le \di (hT^{\bfp_X}, hT^{\bfp_Y})\le \di(T^{\bfp_X}, T^{\bfp_Y}).$$
\end{cor}

\begin{rem}
Compare the leftmost inequality above (for $n=1$) with Lemma \ref{prop_interleaving}.
\end{rem}

\section{Discussion}

We have introduced the persistent homotopy type distance $\dht$ to quantify perturbations of functions defined on homotopy equivalent spaces. As a key consequence,  we were able to lift the standard stability result of persistence diagrams  of sublevel set filtrations of Cohen-Steiner, Edelsbrunner, and Harer \cite{edelsStab}  from the setting of functions defined on the \emph{same} domain to the more general setting of functions defined on possibly different but homotopy equivalent domains.

Focusing, as we did,  on sublevelsets filtrations implies a certain asymmetry in the definition of persistent homotopy type distance in that only up-shifts of maps impact it,  but not down-shifts, see Proposition \ref{prop:retract} and Example \ref{comb}.

 This lack of symmetry is shared neither by  the $L^\infty$ distance nor the natural pseudo-distance. One interesting line of research is whether a certain modification of our definition of the persistent homotopy type distance would correct this asymmetry. Further, it would be interesting to investigate whether such modification implies a lift of the standard $L^\infty$ stability results for extended persistence \cite{ext-pers} and/or interlevel set persistence \cite{inter-pers} to settings when functions are defined on possibly different homotopy equivalent domains. 

Another interesting open line of research is uncovering the relationship, if any, between our homotopy type distance and the distances that appear in the recent work of Blumberg and Lesnick \cite{lesnick-hid}.

\section{Acknowledgements}
This work started during a visit of the third author to the first author at the University of Bologna in 2014 which was partially supported by INdAM-GNSAGA.  The first two authors partially carried out this research within the
activities of ARCES “E. De Castro”, University of Bologna. The third author has been partially supported by NSF under grants DMS-1547357, CCF-1526513, and IIS-1422400. The authors thank Francesca Cagliari for her helpful advice about the categorical setting and Michael Lesnick for useful discussions that helped us to focus Section \ref{sec:interleaving}.


\bibliographystyle{amsplain}

\end{document}